\documentclass[11pt,onecolumn,twoside]{IEEEtran}
\usepackage{amsmath}
\usepackage{amsthm}
\usepackage{comment}  
\usepackage{amsfonts}
\usepackage{graphicx}
\usepackage{subfigure}
\usepackage{psfrag}
\usepackage[normalem]{ulem}
\usepackage{latexsym}
\usepackage{amssymb}
\usepackage{url}

\usepackage{array}
\usepackage{color}
\usepackage{setspace}
\usepackage{times,color}
\usepackage[usenames,dvipsnames,svgnames,table]{xcolor}

\DeclareMathAlphabet{\mathbfsl}{OT1}{ppl}{b}{it} 

\usepackage{tikz}
\usetikzlibrary{shapes,arrows}

\usepackage{ifthen}

\tikzset{
    myarrow/.style={
        draw,
        fill=black,
        single arrow,
        minimum height=1ex,
        line width=0.3pt,
        single arrow head extend=0.1ex
    }
}

\newcommand{\bfc}{{\boldsymbol c}}
\newcommand{\bff}{{\boldsymbol f}}
\newcommand{\bfv}{{\boldsymbol v}}
\newcommand{\bfx}{{\boldsymbol x}}
\newcommand{\bfy}{{\boldsymbol y}}
\newcommand{\bfs}{{\boldsymbol s}}
\newcommand{\bfu}{{\boldsymbol u}}
\newcommand{\whbfc}{\widehat{\bfc}}

\newcommand{\deff}{\mbox{$\stackrel{\rm def}{=}$}}

\makeatletter
\DeclareRobustCommand{\nsbinom}{\genfrac[]\z@{}}
\makeatother

\newcommand{\ceil}[1]{\lceil{#1}\rceil}

\newcommand{\field}[1]{\mathbb{#1}}
\newcommand{\A}{\field{A}}

\newcommand{\F}{\field{F}}

\newcommand{\cA}{{\cal A}}

\newcommand{\cC}{{\cal C}}
\newcommand{\cD}{{\cal D}}
\newcommand{\cF}{{\cal F}}

\newcommand{\cS}{{\cal S}}
\newcommand{\cT}{{\cal T}}
\newcommand{\cP}{{\cal P}}

\newcommand{\C}{{\mathbb C}}

\newcommand{\linadd}{\kern1pt\mbox{\small$\boxplus$}\kern1pt}

\newtheorem{theorem}{Theorem}
\newtheorem{lemma}{Lemma}

\newtheorem{corollary}{Corollary}

\newtheorem{example}{Example}

\newtheorem{proposition}{Proposition}

\theoremstyle{definition}

\newtheorem{definition}{Definition}
\newtheorem{construction}{Construction}

\newcommand{\arrowdown}{%
\tikz [baseline=-1ex]{\node [myarrow,rotate=-90] {};}
}

\begin{document}

\onehalfspacing

\bibliographystyle{plain}

\title{\textbf{Constrained de Bruijn Codes: Properties,\\ Enumeration, Constructions, and Applications}}

\author{
{\sc Yeow Meng Chee}\thanks{Yeow Meng Chee is with the Department of Industrial Systems Engineering
and Management, National University of Singapore, Singapore, e-mail: {\tt pvocym@nus.edu.sg}.
} \hspace{1cm}
{\sc Tuvi Etzion}\thanks{Tuvi Etzion is with Department of Computer Science, Technion,
Haifa 3200003, Israel, e-mail: {\tt etzion@cs.technion.ac.il.} The research of T. Etzion was supported
in part by the BSF-NSF grant no. 2016692 and in part by the ISF grant no. 222/19.
} \hspace{1cm}
{\sc Han Mao Kiah}\thanks{Han Mao Kiah is with School of Physical and Mathematical Sciences, Nanyang Technological Institute,
Singapore, e-mail: {\tt hmkiah@ntu.edu.sg}.
} \\ \vspace{0.5cm}
{\sc Alexander Vardy}\thanks{Alexander Vardy is with Department of Electrical and Computer Engineering,
University of California at San Diego, La Jolla 92093, USA, e-mail: {\tt avardy@ucsd.edu}.
} \hspace{1cm}
{\sc Van Khu Vu}\thanks{Van Khu Vu is with Department of Industrial Systems Engineering
and Management, National University of Singapore, Singapore, e-mail: {\tt isevvk@nus.edu.sg}.
} \hspace{1cm}
{\sc Eitan Yaakobi}\thanks{Eitan Yaakobi is with Department of Computer Science, Technion,
Haifa 3200003, Israel, e-mail: {\tt yaakobi@cs.technion.ac.il.}
\newline Parts of this work were presented in the ISIT2018 and the ISIT2019.
}
}
\maketitle

\begin{abstract}
The \emph{de Bruijn graph}, its sequences, and their various generalizations,
have found many applications in information theory, including many
new ones in the last decade. In this paper, motivated by a coding problem for emerging memory technologies,
a set of sequences which generalize sequences in the de Bruijn graph are defined.
These sequences can be also defined and viewed as constrained sequences.
Hence, they will be called \emph{constrained de Bruijn sequences} and a set of such
sequences will be called a \emph{constrained de Bruijn code}. Several
properties and alternative definitions for such codes are examined and they are analyzed
as generalized sequences in the de Bruijn graph (and its generalization) and
as constrained sequences. Various enumeration techniques are used to compute the total number
of sequences for any given set of parameters. A construction method
of such codes from the theory of shift-register sequences is proposed.
Finally, we show how these constrained de Bruijn sequences and codes can be applied in constructions
of codes for correcting synchronization errors in the $\ell$-symbol read channel and in the racetrack memory channel.
For this purpose, these codes are superior in their size on previously known codes.
\end{abstract}

\vspace{0.5cm}


\vspace{0.5cm}



\newpage
\section{Introduction}
\label{sec:introduction}

The de Bruijn graph of order $m$, $G_m$, was introduced in 1946 by de Bruijn~\cite{deB46}.
His target in introducing this graph was to find a recursive method to enumerate
the number of cyclic binary sequences of length~$2^k$ such that
each binary $k$-tuple appears as a window of length $k$ exactly once in
each sequence. It should be mentioned that
in parallel also Good~\cite{Goo46} defined the same graph and hence it is sometimes called the
\emph{de Bruijn-Good graph}. Moreover, it did not take long for de Bruijn himself to find out
that his discovery was not novel.
In 1894 Flye-Sainte Marie~\cite{Mar94} has proved the enumeration result of de Bruijn without the definition
of the graph. Nevertheless, the graph continues to carry the name of de Bruijn as well as the related sequences
(cycles in the graph) which he enumerated.
Later in 1951 van Aardenne-Ehrenfest and de Bruijn~\cite{AaB51}
generalized the enumeration result for any arbitrary alphabet of finite size $\sigma$ greater than one,
using a generalized graph for an alphabet $\Sigma$ of size~$\sigma$.
Formally, the de Bruijn graph $G_{\sigma,k}$ has $\sigma^k$ vertices, each one is represented by a word
of length~$k$ over an alphabet $\Sigma$ with $\sigma$ letters. The in-degree and the out-degree of
each vertex in the graph is~$\sigma$. There is a directed edge from the vertex $(x_0,x_1,\ldots,x_{k-1})$
to the vertex $(y_1,y_2,\ldots,y_k)$, where $x_i,y_j \in \Sigma$, if and only if $y_i = x_i$, $1 \leq i \leq k-1$.
This edge is represented by the $(k+1)$-tuple $(x_0,x_1,\ldots,x_{k-1},y_k)$.
The sequences enumerated by de Bruijn are those whose length is $\sigma^k$ and each $k$-tuple over $\Sigma$ appears in
a window of $k$ consecutive (cyclically) symbols in the sequence.
Such a sequence enumerated by de Bruijn is represented by an Eulerian cycle
in $G_{\sigma,k-1}$, where each $k$ consecutive symbols represent an edge in the graph.
This sequence is also an Hamiltonian cycle in $G_{\sigma,k}$, where each consecutive
symbols represent a vertex in the graph. Henceforth, we assume that the cycles are Hamiltonian,
i.e., with no repeated vertices.

Throughout the years since de Bruijn introduced his graph and the related sequences, there have been
many generalizations for the graph and for the sequences enumerated by de Bruijn.
Such generalizations include enumeration of sequences of length $\ell$, where $\ell < \sigma^k$,
in $G_{\sigma,k}$~\cite{Mau92} or coding for two-dimensional arrays in which all the $n \times m$ sub-arrays
appear as windows in exactly one position of the large array~\cite{Etz88}.
The interest in the de Bruijn graph, its sequences, and their generalizations, is due
to their diverse important applications. One of the first applications of this
graph was in the introduction of shift-register sequences in general and linear feedback shift
registers in particular~\cite{Gol67}. These will have an important role also in our research.
Throughout the years de Bruijn sequences, the de Bruijn graph, and
their generalizations, e.g. for larger dimensions, have found variety of applications.
These applications include cryptography~\cite{Fre82,Lem79}, and in particular linear complexity
of sequences, e.g.~\cite{CGK82,EKKLP09,EtLe84a,GaCh83,Key76},
interconnection networks, e.g.~\cite{Ben64,EtLe86,SaPr89,Sto71,VaRa88},
VLSI testing, e.g.~\cite{BCR83,LeCo85}, two-dimensional generalizations, e.g.~\cite{BEGGHS,Etz88,McSl76} with applications
to self-locating patterns,
range-finding,  data scrambling, mask configurations,
robust undetectable digital watermarking of two-dimensional test
images, and structured light, e.g.~\cite{Hsi01,MOCDZN,PSCF05,SPB10,vSTO94}.
Some interesting modern applications are combined with biology, like
the genome assembly as part of DNA sequencing, e.g.~\cite{CBP09,CPT11,IdWa95,LiWa03,PTW01,ZhWa03}
and coding for DNA storage, e.g.~\cite{CCEK17,GaMi18a,KPM16,SGSD17}.
This is a small sample of examples for the wide range of applications in which
the de Bruijn graph, its sequences, and their generalizations, were used.

The current work is not different. Motivated by applications to certain coding problems for storage,
such as the $\ell$-symbol read channel and the racetrack memory channel, we introduce
a new type of generalization for de Bruijn sequences, the \emph{constrained de Bruijn sequences}.
In these sequences, a $k$-tuple cannot repeat within a segment starting in any of $b$ consecutive positions.
This generalization is quite natural and as the name hints, it can be viewed as a type of
a constrained sequence. The goal of this paper is to study this type of sequences in all natural
directions: enumeration, constructions, and applications.

Our generalization is motivated by the need to combat synchronization errors (which are shift errors known
also as deletions and sticky insertions) in certain memories.
These types of synchronization errors occur in some new memory technologies, mainly in racetrack memories~\cite{CKVVY17A,CKVVY18},
and in other technologies which can be viewed as an $\ell$-symbol read channel~\cite{CB11,CJKWY13,YBS16}.
By using constrained de Bruijn sequences to construct such codes we will
be able to increase the rate of codes which correct such synchronization errors.
But, we believe that de Bruijn constrained sequences and codes (sets of sequences) are of interest for their own right from
both practical and theoretical points of view.
The new defined sequences can be viewed as constrained codes and as such they pose some interesting problems.
This is the reason that the new sequences and the related codes will be called \emph{constrained
de Bruijn sequences} and \emph{constrained de Bruijn} codes, respectively.

The rest of this paper is organized as follows.
In Section~\ref{sec:prem} we introduce some necessary concepts which are important in our study,
such as the length and the period of sequences in general and of de Bruijn
sequences in particular. We will also introduce some elementary concepts related
to shift-register sequences.
In Section~\ref{sec:properties} we define the new type of sequences, the constrained
de Bruijn sequences and their related codes.
There will be a distinction between cyclic and acyclic sequences. We consider the concept of periodicity for
acyclic and cyclic sequences and define the concept of forbidden patterns.
The main result in this section will be a theorem which reveals that by the given definitions, three
types of codes defined differently, form exactly the same set of sequences for an appropriate
set of parameters for each definition.
In Section~\ref{sec:enumerate} the enumeration for the number
of constrained de Bruijn sequences with a given set of parameters will be considered.
We start with a few very simple enumeration results and
continue to show that for some parameters, the rates of the new defined codes approach 1, and there
are other parameters with one symbol of redundancy. Most of the section will be devoted to a consideration of the new defined
codes as constrained codes. Enumeration based on the theory of constrained codes will be applied.
These considerations yield also efficient encoding and decoding algorithms for the new defined codes.
In Section~\ref{sec:construct}, a construction based on shift-register sequences will be given.
This construction will be our main construction
for codes with large segments in which no repeated $k$-tuples appear.
The next two sections are devoted to applications of constrained de Bruijn sequences in
storage memories. In Section~\ref{sec:symbol_read}, the application to the $\ell$-symbol read channel
is discussed. This application yields another application for another new type of storage channel,
the racetrack channel which is considered in Section~\ref{sec:racetrack}.
We conclude in Section~\ref{sec:conclude}, where we also present a few directions
for future research.

\section{Preliminaries}
\label{sec:prem}

In this section we will give some necessary definitions and notions concerning, cyclic and acyclic sequences,
paths and cycles in the de Bruijn graph, shift-register sequences and in particular those which are related
to primitive polynomials and are known as maximum length shift-register sequences. In Section~\ref{sec:length_period}
we discuss the length and period of sequences as well as cyclic and acyclic sequences. Shift-register sequences are
discussed in Section~\ref{sec:shift_registers}.

\subsection{Paths and Cycles in the de Bruijn Graph}
\label{sec:length_period}

A \emph{sequence} $\bfs = (s_1,s_2,\ldots,s_n)$ over an alphabet $\Sigma$ is a sequence of symbols from $\Sigma$, i.e.,
$s_i \in \Sigma$ for $1 \leq i \leq n$. The \emph{length} of such a sequence
is the number of its symbols $n$ and the sequence is considered to be acyclic.
A \emph{cyclic sequence} $\bfs = [s_1,s_2,\ldots,s_n]$ is a sequence of length $n$ for which the symbols can be read from any position
in a sequential order, where the element $s_1$ follows the element $s_n$. This means that $\bfs$ can be written also as
$[s_i,s_{i+1},\ldots,s_n,s_1,\ldots,s_{i-1}]$ for each $2 \leq i \leq n$.

\begin{definition}
A cyclic sequence $\bfs=[s_1,\ldots,s_n]$, where $s_i \in \Sigma$ and $\sigma = |\Sigma|$,
is called a \emph{weak de Bruijn sequence (cycle)}
of order $k$, if all the $n$ windows of consecutive $k$ symbols of $\bfs$ are distinct.
A cyclic sequence $\bfs=[s_1,\ldots,s_n]$ is called a \emph{de Bruijn sequence}
of order $k$, if $n=\sigma^k$ and $\bfs$ is a weak de Bruijn sequence of order $k$.
\end{definition}

The connection between a weak de Bruijn cycle of length $n$ (as a cycle in the de Bruijn graph) and a weak de Bruijn sequence
of length $n$ in the graph is very simple. The sequence is generated from the cycle by considering the first digit
of the consecutive vertices in the cycle. The cycle is generated from the sequence by considering the consecutive
windows of length $k$ in the sequence. This is the place to define some common concepts and some
notation for sequences. An acyclic sequence $\bfs=(s_1,\ldots,s_n)$ has $n$ digits read from the first
to the last with no wrap around. If all the windows of length $k$, starting at position $i$, $1 \leq i \leq n-k+1$,
are distinct, then the sequence corresponds to a simple path in~$G_{\sigma,k}$. The \emph{period} of a cyclic sequence
$\bfs = [s_0,s_1,\ldots,s_{n-1}]$ is the least integer $p$, such that $s_i = s_{i+p}$,
where indices are taken modulo $n$, for each $i$, $0 \leq i \leq n-1$.
It is well known that the period $p$ divides the length of a sequence $n$. If $n= rp$ then the periodic sequence $\bfs$ is
formed by a concatenation of $r$ copies of the first $p$ entries of $\bfs$. It is quite convenient to have the
length $n$ and the period $p$ equal if possible. There is a similar definition for a period of
acyclic sequence with slightly different properties, which is given in Section~\ref{sec:properties}.
In this paper we will assume (if possible) that for a given cyclic sequence $\bfs=[s_1,\ldots,s_n]$ the period is $n$ as the length.
Hence, usually we won't distinguish between the length and period for cyclic sequences. We will elaborate
more on this point in Section~\ref{sec:shift_registers}. Finally, the substring (window)
$(s_i,s_{i+1},\ldots,s_j)$ will be denoted by $s[i,j]$ and this substring will be always considered
as an acyclic sequence, no matter if $\bfs$ is cyclic or acyclic.
In addition, $s[i]$ will be sometimes used instead of $s_i$ and the set $\{ 1,2,\ldots,n \}$ is denoted by $[n]$.

\subsection{Feedback Shift Registers and their Cycle Structure}
\label{sec:shift_registers}

The theory on the sequences in the de Bruijn graph cannot be separated from the
theory of shift-register sequences developed mainly by Golomb~\cite{Gol67}.
This theory, developed fifty years ago,
was very influential in various applications related to digital communication~\cite{Gol82,GoGo05}.
A short summary on the theory of shift-registers taken from~\cite{Gol67} which is related to our work,
is given next.

A \emph{characteristic polynomial} (for a linear feedback shift register
defined in the sequel) $c(x)$ of degree~$k$ over~$\F_q$, the finite field with $q$ elements, is a polynomial given by
$$
c(x) = 1 - \sum_{i=1}^k c_i x^i~,
$$
where $c_i \in \F_q$. For such a polynomial a function $f$ on $k$ variables from $\F_q$ is defined by
$$
f(x_1,x_2,\ldots,x_k)=\sum_{i=1}^k c_i x_{k+1-i}~.
$$
For the function $f$ we define a \emph{state diagram} with the set of vertices
$$
Q^k \triangleq \{ (x_1,x_2,\ldots,x_k) ~:~ x_i \in \F_q \}~.
$$
If $x_{k+1} = f(x_1,x_2,\ldots,x_k)$ then an edge from $(x_1,x_2,\ldots,x_k)$ to
$(x_2,\ldots,x_k,x_{k+1})$ is defined for the related state diagram.
A \emph{feedback shift register} of length $k$ has $q^k$ states corresponding to the set~$Q^k$ of
all $q^k$~$k$-tuples over $\F_q$. The feedback function $x_{k+1}=f(x_1,x_2,\ldots,x_k)$ defines
a mapping from $Q^k$ to $Q^k$. The feedback function can be linear or nonlinear.
The shift register is called \emph{nonsingular} if its state diagram consists of disjoint
cycles. Any such state diagram is called a \emph{factor} in $G_{q,k}$, where
a factor in a graph is a set of vertex-disjoint cycles which contains
all the vertices of the graph. There is a one-to-one correspondence between
the set of factors in the de Bruijn graph $G_{q,k}$ and the set of state diagrams
for the nonsingular shift registers of order $k$ over $\F_q$.
It is well-known~\cite{Gol67} that a binary feedback shift-register
is nonsingular if and only if its feedback function has the form
$$
f(x_1,x_2,\ldots,x_k) = x_1 + g(x_2,\ldots,x_k)~,
$$
where $g(x_2,\ldots,x_k)$ is any binary function on $k-1$ variables.
A similar representation also exists for nonsingular feedback shift-registers over $\F_q$.

An Hamiltonian cycle in $G_{q,k}$ is a de Bruijn cycle which forms a de Bruijn sequence.
There are $(q!)^{q^{k-1}}/q^k$ distinct such cycles in $G_{q,k}$ and there are many methods to generate
such cycles~\cite{EtLe84,Fre82}. One important class of sequences in the graph, related to de Bruijn sequences, are the so called
m-sequences, or \emph{maximal length linear shift-register sequences}. A shift-register is called \emph{linear}
if its feedback function $f(x_1,x_2,\ldots,x_k)$ is linear. An \emph{m-sequence} is a
sequence of length $q^k-1$ generated by a linear shift-register associated
with a primitive polynomial of degree $k$ over $\F_q$. Each primitive polynomial is associated
with such a sequence. In such a sequence all windows of length $k$ are distinct and the only
one which does not appear is the all-zero window. The following theorem is well known~\cite{Gol67}.

\begin{theorem}
\label{thm:num_primitive}
The number of distinct m-sequences of order $k$ over $\F_q$ (the same as the number of primitive polynomials
of order $k$ over $\F_q$) is
$$
\frac{\phi(q^k-1)}{k},
$$
where $\phi$ is the Euler function.
\end{theorem}

The \emph{exponent} of a polynomial $f(x)$ is the smallest integer $e$
such that $f(x)$ divides $x^e -1$. The length of the longest cycle, of the state diagram, formed
by the shift-register associated with the characteristic polynomial $f(x)$ is the exponent
of $f(x)$. The length of the cycles associated with a multiplication of several
irreducible polynomials can be derived using the exponents of these polynomials and some related algebraic
and combinatorial methods. This theory leads immediately to the following important result.

\begin{theorem}
\label{thm:multiply_two_primitive}
The state diagram associated with a multiplication of two distinct characteristic primitive
polynomials $f(x)$ and $g(x)$ of order $k$ over $\F_q$ contains $q^k+1$ cycles
of length $q^k-1$ and the all-zero cycle. Each possible $(2k)$-tuple over $\F_q$
appears exatly once as a window in one of these cycles. The m-sequences related to $f(x)$
and $g(x)$ are two of these cycles.
\end{theorem}

\section{Constrained de Bruijn Codes, Periods, and Forbidden Patterns}
\label{sec:properties}

In this section we give the formal definitions for constrained de Bruijn sequences and codes.
We present a definition for a family of sequences with a certain period and a definition
for a family of sequences avoiding certain substrings. Three different definitions will be given and
it will be proved that with the appropriate parameters the related three families of sequences contain
the same sequences.
Each definition will have later some role in the enumeration of constrained de Bruijn sequences which
will be given in Section~\ref{sec:enumerate}.

\subsection{Constrained de Bruijn Codes}

\begin{definition}
$~$
\begin{itemize}
\item A sequence $\bfs=(s_1,\ldots,s_n)$, over some alphabet $\Sigma$,
is called a \textbf{\emph{$(b,k)$-constrained de Bruijn sequence}} if
$\bfs[i,i+k-1] \neq \bfs[j,j+k-1]$ for all $i,j \in [n-k+1]$ such that $0< |i-j| \leq b-1$.
In other words, in each substring of $\bfs$ whose length is $b+k-1$ there is no repeated $k$-tuple, i.e.
each subset of $b$ consecutive windows of length $k$ in $\bfs$ contains $b$ distinct $k$-tuples.

\item A set of distinct $(b,k)$-constrained de Bruijn sequences of length~$n$ is called
a \textbf{\emph{$(b,k)$-constrained de Bruijn code}}. The set of all
$(b,k)$-constrained de Bruijn sequences of length~$n$ will be denoted by $\C_{DB}(n,b,k)$.
The alphabet $\Sigma$ and its size $sigma$ should be understood from the context throught
our discussion.

\item A cyclic sequence $\bfs=[s_1,\ldots,s_n]$, over $\Sigma$,
is called a \textbf{\emph{cyclic $(b,k)$-constrained de Bruijn sequence}} if
$(s_1,\ldots,s_n)$ is a $(b,k)$-constrained de Bruijn sequence and
$s[i,i+k-1] \neq s[j,j+k-1]$ for all $i,j \in [n]$ such that $i <j$ and $i-j+n \leq b-1$.
In other words, in each cyclic substring of $\bfs$ whose length is $b+k-1$ there is no repeated $k$-tuple.
Note, that two sequences which differ only in a cyclic shift are considered to be the
same sequence.

\item A set of distinct cyclic $(b,k)$-constrained de Bruijn sequences of length~$n$ is called
a \textbf{\emph{cyclic $(b,k)$-constrained de Bruijn code}}. The set of all cyclic
$(b,k)$-constrained de Bruijn sequences of length~$n$ will be denoted by $\C^*_{DB}(n,b,k)$.
We note that by the definition of a cyclic constrained de Bruijn sequence two codewords
in a cyclic $(b,k)$-constrained de Bruijn code
cannot differ only in a cyclic shift, but in some applications one might consider these cyclic
shifts of codewords as distinct codewords.
\end{itemize}
\end{definition}

de Bruijn sequences form a special case of constrained de Bruijn sequences as asserted in the
following theorem which is readily verified by definition.

\begin{theorem}
\label{thm:eqDB}
$~$
\begin{itemize}
\item The cyclic sequence $\bfs$ of length $q^k$ over $\F_q$ is a de Bruijn sequence if and only if $\bfs$
is a cyclic $(q^k,k)$-constrained de Bruijn sequence.

\item The (acyclic) de Bruijn sequence $\bfs$ of length $q^k+k-1$ over $\F_q$ is a de Bruijn sequence
if and only if $\bfs$ is a $(q^k,k)$-constrained de Bruijn sequence.
\end{itemize}
\end{theorem}

Theorem~\ref{thm:eqDB} introduce the connection between
de Bruijn sequences and constrained de Bruijn sequences.
But, there are a few main differences between de Bruijn sequences of order $k$ and $(b,k)$-constrained
de Bruijn sequences.

\begin{itemize}
\item A de Bruijn sequence is usually considered and is used as a cyclic sequence while a constrained de Bruijn sequence
will be generally used as an acyclic sequence.

\item A de Bruijn sequence contains each possible $k$-tuple exactly once as a window of length~$k$,
while in a $(b,k)$-constrained de Bruijn sequence, each possible $k$-tuple can be repeated several times
or might not appear at all.

\item A de Bruijn sequence contains each possible $k$-tuple exactly once in one period of
the sequence, while in a $(b,k)$-constrained de Bruijn sequence, each possible $k$-tuple can appear at most
once among any $b$ consecutive $k$-tuples in the sequence.

\item The length (or equivalently period in this case) of a de Bruijn sequence is strictly $q^k$,
while there is no constraint on the length and  the period of a $(b,k)$-constrained de Bruijn sequence.
\end{itemize}

These differences between de Bruijn sequences and constrained de Bruijn sequences
are important in the characterization, enumeration, constructions, and
applications of the two types of sequences.

\subsection{Periods and Forbidden Subsequences}

There are two concepts which are closely related to constrained de Bruijn sequences,
the \emph{period} of an acyclic sequence and \emph{avoided} patterns. Let $\bfu$ and $\bfs$
be two sequences over $\Sigma$. The sequence $\bfs$ \emph{avoids} $\bfu$ (or $\bfs$ is\linebreak \emph{$\bfu$-avoiding})
if $\bfu$ is not a substring of $\bfs$. Let $\cF$ be a set of sequences over $\Sigma$ and $\bfs$ a sequence
over $\Sigma$. The sequence $\bfs$ \emph{avoids} $\cF$ (or $\bfs$ is \emph{$\cF$-avoiding}) if no sequence
in $\cF$ is a substring of $\bfs$. Let $\cA(n;\cF)$ denote the set of all $\sigma$-ary sequences
of length $n$ which avoid $\cF$. A subset of $\cA(n;\cF)$, i.e.
a set of $\cF$-avoiding sequences of length~$n$, is called an \emph{$\cF$-avoiding code} of length~$n$.

The second concept is the \emph{period} of a sequence. For cyclic sequences the length and the period
of the sequence either coincide or have a very strict relation, where the period
of the sequence divides its length. This is not the case for acyclic sequences.

\begin{definition}
$~$
\begin{itemize}
\item A sequence $\bfs = (s_1,s_2,\ldots,s_n)\in \Sigma^n$ is a \textbf{\emph{period $p$ sequence}} if it satisfies
$s_i=s_{i+p}$ for all $1\leq i\leq n-p$. Note, that the definition for periods of a cyclic sequence coincides with
this definition).

\item A sequence $\bfs\in \Sigma^n$ is called
an \textbf{\emph{$m$-limited length for period $p$ substrings}} if any substring with period $p$ of $\bfs$ has length at most~$m$.

\item A set of $m$-limited length, period $p$ sequences from $\Sigma^n$  is called
a \textbf{code of \emph{$m$-limited length for period $p$ substrings}}. The set of all
such $m$-limited length for period $p$ substrings is denoted by $\C_{LP}(n,m,p)$.
\end{itemize}
\end{definition}

The first lemma is a straightforward observation.
\begin{lemma}
\label{lem:lengthm+1avoid}
If $\cF$ is the set of all sequences of length $m+1$ and period $p$, then for each $i$, $1 \leq i \leq m$,
the set $\C_{LP}(n,i,p)$ is an $\cF$-avoiding code.
\end{lemma}
\begin{proof}
Recall, that for each $i$, the set of sequences in $\C_{LP}(n,i,p)$, have length $n$.
If $\bfs$ is a sequence in $\C_{LP}(n,i,p)$, $1 \leq i \leq m$, then each substring of $\bfs$ with period $p$ has length
at most $i$, where $i$ is a positive integer strictly smaller than $m+1$.
The set $\cF$ contains the sequences of length $m+1$ and period $p$ and hence
a sequence of $\cF$ cannot be a substring of $\bfs$. Thus, $\C_{LP}(n,i,p)$ is an $\cF$-avoiding code.
\end{proof}

\subsection{Equivalence between the Three Types of Codes}

Let $\cF_{p,p+k}$ be the set of all period $p$ sequences of length $p+k$ for any
given $1 \leq p \leq b-1$ and let $\cF= \cup_{p=1}^{b-1} \cF_{p,p+k}$.
The following result implies a strong relation between constrained de Bruijn codes,
code of limited length for period $p$ substrings, and $\cF$-avoiding codes.
By the related definitions of these concepts we have the following theorem.
\begin{theorem}
\label{thm:relation}
For all given admissible $n,~b,~k$, and any code $\C \subset \Sigma^n$,
$$
\cA(n;\cF)=\C_{DB}(n,b,k)= \bigcap_{i=1}^{b-1}\C_{LP}(n,i+k-1,i)~.
$$
\end{theorem}
\begin{proof}
We will prove that $\cA(n;\cF) \subseteq \C_{DB}(n,b,k)$,
$\C_{DB}(n,b,k) \subseteq \bigcap_{i=1}^{b-1}\C_{LP}(n,i+k-1,i)$, and
$\bigcap_{i=1}^{b-1}\C_{LP}(n,i+k-1,i) \subseteq \cA(n;\cF)$.
These three containment proofs will imply the claim of the theorem.

First, we prove that if $\bfc$ is a sequence in $\cA(n;\cF)$ then $\bfc$ is a sequence in $\C_{DB}(n,b,k)$.
Let $\bfc=(c_1,c_2,\ldots,c_n)$ be any sequence in $\cA(n,\cF)$, and assume for
the contrary that $\bfc \not \in \C_{DB}(n,b,k)$.
This implies that there exist integers $i,j \in [n-k+1]$ such that $1 \leq p = j-i \leq b-1$ and
$\bfc[i,i+k-1] = \bfc[j,j+k-1]$. Hence, $\bfc[i,j+k-1]=(c_i,c_{i+1},\ldots,c_{j+k-1})$ is a substring
with period $p$ and length $p+k$. Therefore $\bfc[i,j+k-1] \in \cF$, a
contradiction since $\bfc \in \cA(n;\cF)$. Thus, $\bfc \in \C_{DB}(n,b,k)$ which implies that $\cA(n;\cF) \subseteq \C_{DB}(n,b,k)$

Next, we prove that if $\bfc$ is a sequence in $\C_{DB}(n,b,k)$ then $\bfc$ is a sequence
of $\bigcap_{i=1}^{b-1}\C_{LP}(n,i+k-1,i)$. Let ${\bfc=(c_1,c_2,\ldots,c_n)}$ be any
sequence in $\C_{DB}(n,b,k)$, and assume for the contrary that
$\bfc \not \in \bigcap_{i=1}^{b-1}\C_{LP}(n,i+k-1,i)$. Hence, there exists $p \in [b-1]$
such that $\bfc \not \in \C_{LP}(n,p+k-1,p)$. Therefore, $\bfc$ contains a period $p$
substring of length $p+k$. Let $\bfc[i,i+p+k-1]$ be such a period $p$ substring.
Hence, $\bfc[i,i+k-1]=\bfc[i+p,i+p+k-1]$, a contradiction since $\bfc \in \C_{DB}(n,b,k)$.
Thus, $\C \subseteq \bigcap_{i=1}^{b-1}\C_{LP}(n,i+k-1,i)$ which implies that $\C_{DB}(n,b,k) \subseteq \bigcap_{i=1}^{b-1}\C_{LP}(n,i+k-1,i)$.

Finally, if $\bfc$ is a sequence in $\bigcap_{i=1}^{b-1}\C_{LP}(n,i+k-1,i)$,
then by Lemma~\ref{lem:lengthm+1avoid} we have that $\bfc$ is a sequence in $\cA(n;\cF)$.
Thus, $\bigcap_{i=1}^{b-1}\C_{LP}(n,i+k-1,i) \subseteq \cA(n;\cF)$.
\end{proof}

\begin{corollary}
\label{cor:relation}
For all given admissible $n,b,k$, and any code ${\C \subset \Sigma^n}$,
the following three statements are equivalent
\begin{enumerate}
\item $\C$ is a subset of  $\cA(n;\cF).$
\item $\C$ is a subset of  $\C_{DB}(n,b,k)$.
\item $\C$ is a subset of  $\bigcap_{i=1}^{b-1}\C_{LP}(n,i+k-1,i)$.
\end{enumerate}
\end{corollary}

The set $\cF$ of all forbidden patterns of an $\cF$-avoiding code $\cC$ is never a minimal minimal, i.e.,
there exists another set $\cF'$ such that $\cF' \subset \cF$ and $\cC$ is also
an $\cF'$-avoiding code. It implies that
there always exist two different sets, $\cF_1$ and $\cF_2$, where $\cF_1 \subset \cF_2$, such that $\cA(n;\cF_1) = \cA(n;\cF_2)$.
But, there always exists one such set $\cF$ whose size is smaller than the sizes of all the other sets
To see that, let $\cF$ be a set of forbidden sequences over $\Sigma$. As long as $\cF$ contains
$\sigma$ sequences of length $\ell +1$ which form the set $S = \{ \bfu \alpha ~:~ \alpha \in \Sigma \}$,
where $\bfu$ is a string of length~$\ell$, these $\sigma$ sequences of $\cF$ which are
contained in $S$ can be replaced by the sequence~$\bfu$. When this process comes to its end, instead of the
original set $\cF$ of forbidden patterns we have a set $\cF'$ of forbidden patterns for the same code.
The set of sequences~$\cF'$ will be called the \emph{forbidden reduced set} of $\cF$.

\begin{lemma}
\label{lem:reduce_F}
If $\cF$ is a set of forbidden sequences and $\cF'$ is its forbidden reduced set, then
$\cA(n;\cF) = \cA(n;\cF')$.
\end{lemma}
\begin{proof}
The proof is an immediate observation from the fact that all the sequences of length $n$
which do not contain the patterns in $S=\{ \bfu \alpha ~:~ \alpha \in \Sigma \}$, where $n$ is greater than
the length of~$\bfu$, as substrings do not contain the pattern $\bfu$ as a substring.
Hence, $\cA(n;\cF) \subseteq \cA(n;\cF')$.

Clearly, all the sequences of length $n$ which do not contain $\bfu$ as a substring do not contain
any pattern from $S$ as a substring. Hence, $\cA(n;\cF') \subseteq \cA(n;\cF)$.

Thus, $\cA(n;\cF) = \cA(n;\cF')$.
\end{proof}

\section{Enumeration of Constrained de Bruijn Sequences}
\label{sec:enumerate}

In this section we consider enumeration of the number of sequences in $\C_{DB}(n,b,k)$.
Note, that we are considering only acyclic sequences.
A $\sigma$-ary code $\C$ of length~$n$ is a set of $\sigma$-ary sequences of length~$n$,
that is $\C \subseteq \Sigma^n$. For each code $\C$ of length~$n$,
we define the rate of the code $\C$ to be $R(\C)=\log_\sigma (|\C|) / n$,
and the redundancy of the code $\C$ to be $r(\C)=n-\log_\sigma (|\C|)$, where $|\C|$ is the size of the code $\C$.
We define the \emph{maximum asymptotic rate} of $(b,k)$-constrained de Bruijn codes to be
$R_{DB}(b,k) = \limsup_{n \to \infty} \frac{\log_{\sigma} |\C_{DB}(n,b,k)|}{n}$.\footnote{The $\limsup$ can indeed be replaced by a proper $\lim$ \cite{MRS01}.}

\subsection{Trivial Bounds}

In this section trivial bounds on the number of constrained de Bruijn sequences are considered.
The number of (cyclic) de Bruijn sequences
of length $\sigma^n$ over an alphabet of size $\sigma$ is $(\sigma !)^{\sigma^{n-1}}/\sigma^n$~\cite{AaB51} which implies
the following simple result (for acyclic sequences).

\begin{theorem}
\label{thm:from_DB}
For any given a positive integer $\sigma \geq 2$, a positive integer $k$, and $n \geq \sigma^k + k-1$,
$$
|\C_{DB}(n,\sigma^k,k)| = (\sigma!)^{\sigma^{n-1}}~.
$$
\end{theorem}
\begin{corollary}
\label{cor:from_DB}
For any given a positive integer $\sigma \geq 2$ and a positive integer $k$,
$$
R_{DB}(\sigma^k,k)=0~.
$$
\end{corollary}

Corollary~\ref{cor:from_DB} can be generalized for $R_{DB}(b,k)$, where $\sigma^k \geq b \geq b_0$, where
$b_0$ is an integer whose value depends on $\sigma$ and $k$.The reason that the rates are zeroes for so many
values is that once we have a long simple path of length $b$ in $G_{\sigma,k}$ (and the number of such
paths is very large~\cite{Mau92}), to continue the path for a long sequence which is a $(b,k)$-constrained de
Bruijn sequence, the sequence will be almost periodic (with a possibility of some small local changes.
In the case of Theorem~\ref{thm:from_DB}, where the path is a de Bruijn sequence, there is only one way to
continue the path without violating the constraint. There are some intriguing questions in this context. The first one, question to be
is to be specific in the value of $b_0$.
Another question is to find a good bound on $R_{DB}(b,k)$,
where $b$ is large compared to $k$ and
$R_{DB}(b,k) > 0$. Such a bound will be given in Section~\ref{sec:construct}, but we have no indication how good it is.
The other extreme case is when $b=1$ and hence the sequence is not constrained and we have the following trivial result.
\begin{theorem}
\label{thm:trivial}
For any given positive integer $\sigma \geq 2$, a positive integer $k$, and $n \geq k$,
$|\C_{DB}(n,1,k)| = \sigma^n$ and $R_{DB}(1,k)=1$.
\end{theorem}
Except for these cases, to find the rates of constrained de Bruijn sequences, where $b > 1$, is not
an easy task. When the rate is approaching one, we are interested in the redundancy of $\C_{DB}(n,b,k)$.
Unfortunately, finding the redundancy of $\C_{DB}(n,b,k)$ is even more difficult than to find the rate.
Fortunately, for small values of $b$ we can use the theory of constrained coding, and for $k$ large
enough we can even show that the redundancy is one symbol.

The last trivial case is the $(b,1)$-constrained de Bruijn codes. The sequences of such a code
will be used for error correction of synchronization errors in racetrack memories
in Section~\ref{sec:no_more_heads}. Since this constraint
implies that any $b$ consecutive elements in the sequence will be distinct, we must have that $\sigma \geq b$
to obtain any valid sequence. If $\sigma \geq b$, then we have to choose $b$ elements from the $\sigma$ alphabet letters
to start the sequence and in any other position we can choose any of the alphabet letters, except for the
previous $b-1$ positions in the sequence. Hence, we have
\begin{theorem}
\label{thm:k=1}
For any $b \geq 1$ and any alphabet of size $\sigma$ we have
$$
|\C_{DB}(n,b,1)|=
  \begin{cases}
    0 & \mbox{if $b> \sigma$}  \\
    \binom{\sigma}{b} b! (\sigma-b+1)^{n-b} & \mbox{if $\sigma \geq b$}
  \end{cases}~.
$$
\end{theorem}
\begin{corollary}
For any $b \geq 1$ and any alphabet $\sigma$ we have that $R_{DB}(b,1)=0$ if $\sigma \leq b$ and
$R_{DB}(b,1)=\log_\sigma (\sigma-b+1)$ if $\sigma > b$.
\end{corollary}

\subsection{$(b,k)$-Constrained de Bruijn Codes with Redundancy 1}

We observe that the asymptotic rate $R_{DB}(b,k)$ is getting close to 1 when $b$ is fixed and $k$ tends to infinity.
In this case, we are interested in the redundancy of the code.

For a subset $A$ of a set $U$, the complement of $A$, $A^c$, is the subset which consists
of all the elements of $U$ which are not contained in $A$. For a code $\C$ of length $n$ over $\Sigma$,
$\C^c = \Sigma^n \setminus \C$. The following simple lemma will
be used in the next theorem.

\begin{lemma}
\label{lem:trivial_set}
If $A_i,~1 \leq i \leq m$, are $m$ sets over the same domain then
$$
\left| \left( \cap_{i=1}^m A_i \right)^c \right| \leq \sum_{i=1}^m | A_i^c |.
$$
\end{lemma}
\begin{proof}
It is well known from set theory that
$$
\left( \cap_{i=1}^m A_i \right)^c = \cup_{i=1}^m A_i^c ~.
$$
Combining this with the trivial assertion that for any two sets $A$ and $B$, $|A \cup B| \leq |A| + |B|$ we have that
$$
| \left( \cap_{i=1}^m A_i \right)^c | = | \cup_{i=1}^m A_i^c |  \leq \sum_{i=1}^m | A_i^c |.
$$
\end{proof}

\begin{theorem}
\label{thm:sizelargeh}
For all $\sigma,n$ and $b < k$,\vspace{-1ex}
$$|\C_{DB}(n,b,k) | \geq  \sigma^n\left(1- (b-1) n \cdot \left( \frac{1}{\sigma} \right)^k   \right).\vspace{-1ex}$$
In particular, for $k \geq \ceil{\log_\sigma n + \log_\sigma (b-1)} +1$, the redundancy of $\C_{DB}(n,b,k)$ is at most a single symbol.
\end{theorem}

\begin{proof}
By Theorem~\ref{thm:relation}, $\C_{DB}(n,b,k)=  \bigcap_{i=1}^{b-1}\C_{LP}(n,i+k-1,i)$. Combining this fact with
Lemma~\ref{lem:trivial_set} we have that
\begin{equation}
\label{eq:fr_th}
|\C_{DB}(n,b,k)|=\sigma^n- |\C_{DB}(n,b,k)^c|=\sigma^n-|\left(\cap_{i=1}^{b-1}\C_{LP}(n,i+k-1,i)\right)^c| \geq \sigma^n-\sum_{i=1}^{b-1} | \C_{LP}(n,i+k-1,i)^c|.
\end{equation}
For each $i$, $1\leq i \leq b-1$, a word $\bfc$ is contained in $\C_{LP}(n,i+k-1,i)^c$ if and only if
it has length~$n$ and a substring of period $i$ whose length is at least $i+k$.
There are at most $n-(i+k-1)$ possible starting
positions for such a substring. Once the first $i$ symbols in this substring, of length $i+k$, are known,
its other $k$ symbols are uniquely determined. The remaining $n-(i+k)$ symbols in the word of length $n$
can be chosen arbitrarily (note, that some choices cause other substrings of the same period and larger length
in this word. They might cause other substrings with larger period than $i$, so the computation which follows have many
repetitions and some sequences which do not have to be computed.).
Hence, the number of words in $\C_{LP}(n,i+k-1,i)^c$ is upper bounded by
$$
|\left(\C_{LP}(n,i+k-1,i)\right)^c| \leq (n-i-k+1) \cdot \sigma^i \cdot \sigma^{n-(i+k)} < \sigma^n\cdot n \cdot \left( \frac{1}{\sigma} \right)^k.
$$
Therefore,
$$
\sum_{i=1}^{b-1} |\left(\C_{LP}(n,i+k-1,i)\right)^c| \leq (b-1) \cdot \sigma^n \cdot n \cdot \left( \frac{1}{\sigma} \right)^k,
$$
and hence using (\ref{eq:fr_th}) we have,
$$
|\C_{DB}(n,b,k)| \geq \sigma^n-  (b-1) \cdot \sigma^n \cdot n \cdot \left( \frac{1}{\sigma} \right)^k
= \sigma^n \left( 1- (b-1) \cdot n \cdot \left( \frac{1}{\sigma} \right)^k \right).
$$
In particular, for $k \geq \ceil{\log_\sigma (n)+\log_\sigma (b-1)}+1$, we obtain
$$
|\C_{DB}(n,b,k)| \geq \sigma^n \left( 1- (b-1) \cdot n \cdot \left( \frac{1}{\sigma} \right)^{\log_\sigma n + \log_\sigma (b-1)+1} \right)
\geq \sigma^n \cdot \frac{\sigma-1}{\sigma} \geq (\sigma-1)\cdot \sigma^{n-1} .
$$
Thus, the redundancy of $\C_{DB}(n, b,k)$ is at most a single symbol.
\end{proof}

A weaker bound than the one in Theorem~\ref{thm:sizelargeh} for $\sigma=2$ was given in~\cite{CKVVY18} (Theorem 13).
Finally, to encode the $(b,k)$-constrained de Bruijn code efficiently with only a single symbol
of redundancy, we may use \emph{sequence replacement techniques}~\cite{WiIm10}.
\vspace{-0.4cm}
\subsection{Representation as a Graph Problem}
\label{sec:represent}
\vspace{-0.1cm}
The key for enumeration of the number of $(b,k)$-constrained de Bruijn sequences of length $n$
is a representation of the enumeration as a graph problem. The de Bruijn graph $G_{\sigma,k}$ was
the key for the enumeration of the cyclic (as implies also the acyclic) de Bruijn sequences of
length~$\sigma^k$, i.e., the number of $(\sigma^k,k)$-constrained de Bruijn sequences of length $n \geq \sigma^k +k-1$.
This number is equal to the number of Hamiltonian cycles in the graph and also the number of Eulerian cycles
in $G_{\sigma,k-1}$. This number was found for example with a reduction to the number of spanning trees
in the graph~\cite{Fre82}. The (not necessary simple) paths in $G_{\sigma,k}$ (where vertices
are considered as the $k$-tuples) represent the $(1,k)$-constrained de Bruijn
sequences, i.e. sequences with no constraints. Can the de Bruijn graph or a slight modification of it
represents other constraints. The answer is definitely yes. If we remove from $G_{\sigma,k}$ the $\sigma$ self-loops,
then the paths (represented by the vertices) in the new graph represents all the $(2,k)$-consrained de Bruijn sequences.
This can be generalized to obtain a graph for the $(3,k)$-constrained de Bruijn sequences
and in general the $(b,k)$-constrained de Bruijn sequences, when $b$ is a fixed integer.
For the $(3,k)$ constraint, we form from
$G_{\sigma,k}$ a new graph $G'_{\sigma,k}$ as follows. First, we remove all the self-loops from $G_{\sigma,k}$ to
accommodate the $(2,k)$ constraint. Next, we consider all the cycles of length two in $G_{\sigma,k}$.
These cycles have the form $x \rightarrow y \rightarrow x$, where $x=(x_1,x_2,\ldots,x_k)$ and $y=(y_1,y_2,\ldots,y_k)$,
$x_i=y_j= a$, for odd $i$ and even $j$, and $x_i=y_j= b$, for even $i$ and odd $j$, where $a$ and $b$ are two distinct
letters in $\Sigma$. In $G'_{\sigma,k}$, $x$ and $y$ are replaced by four vertices $x1$, $x2$ instead of $x$, and $y1$, $y2$ instead of $y$.
All the in-edges and the out-edges to and from $x$ (except for the ones from and to $y$)
are duplicated and are also in-edges and out-edges to and from $x1$ and $x2$.
Similarly, all the in-edges and the out-edges in and from $y$ (except for the ones from and to $x$)
are duplicated and are also in-edges and out-edges to and from $y1$ and $y2$.
There is also one edge from $x1$ to $y1$ and from $y2$ to $x2$. All the paths (related to the vertices) in the constructed
graph $G'_{\sigma,k}$ form all the $(3,k)$-constrained de Bruijn sequences. The total number of cycles of length two in
$G_{\sigma,k}$ is $\binom{\sigma}{2}$ and hence the number of vertices in $G'_{\sigma,k}$
is $\sigma^k + \sigma (\sigma -1)$. An example for the representation of the graph (i.e., state diagram)
for the $(3,3)$-constrained de Bruijn sequences via the de Bruijn graph is presented in Figure~\ref{Fig1}.

\begin{figure}
\centering
\begin{tikzpicture}[->,>=stealth',shorten >=1pt,auto,node distance=2cm,
                    thin,main node/.style={circle,draw,font=\sffamily\Large\bfseries}]

  \node[main node] (1) {000};
  \node[main node] (2) [below left of=1] {001};
  \node[main node] (3) [below right of=1] {100};
  \node[main node] (4) [below of=2] {010A};
  \node[main node] (5) [below of=3] {010B};
   \node[main node] (6) [ below of=4] {101B};
     \node[main node] (7) [ below of=5] {101A};
      \node[main node] (8) [ below of=6] {011};
      \node[main node] (9) [ below of=7] {110};
      \node[main node] (10) [ below left of=9] {111};

  \path[every node/.style={font=\sffamily\small}]
       (1) edge node [right] {$1$} (2)
       (2) edge node [right] {$0$} (4)
       (2) edge  [bend right] node[left] {$1$} (8)
       (3) edge node [below] {$1$} (2)
       (3) edge node [right] {$0$} (1)
       (4) edge node [right] {$0$} (3)
       (4) edge node [right] {$1$} (6)
		(5) edge node [right] {$0$} (3)
		(6) edge node [right] {$1$} (8)
		(7) edge node [right] {$0$} (5)
		(7) edge node [right] {$1$} (8)
		(8) edge node [left] {$1$} (10)
		(8) edge node [below] {$0$} (9)
		(10) edge node [right] {$0$} (9)
		(9) edge node [right] {$1$} (7)

      (9)  edge [bend right] node[right] {$0$} (3);
\end{tikzpicture}
\caption{Graph representation for the $(3,3)$-constrained de Bruijn sequences, via the de Bruijn graph\vspace{-3ex}}
\label{Fig1}
\end{figure}
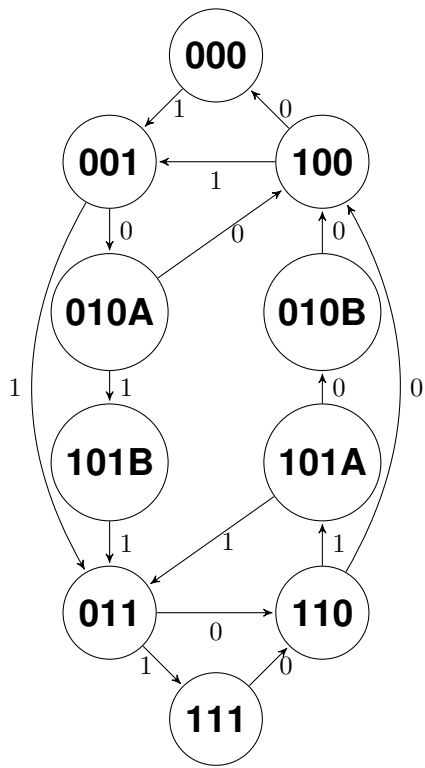

A second representation is by using the theory of constrained coding~\cite{MRS01}.
One have to construct the state diagram of the constrained system. From the
state diagram one has to generate the related adjacency matrix to compute the rate of the
constrained de Bruijn code by computing the largest eigenvalue of the adjacency matrix.
If $\lambda$ is the largest eigenvalue of the adjacency matrix then asymptotically there are
$\lambda^n$ sequences and the rate of the code is $\log_\sigma \lambda$. An example will
be given in the next subsection.

While the first representation is related to the code $\C_{DB}(n, b,k)$, the second representation is related to
the code $\cA(n;\cF)$, where $\cF$ are all the substrings which are forbidden in a $(b,k)$-constrained sequence.
By Theorem~\ref{thm:relation} these two codes are equal.
Each vertex in the state diagram (of the constrained code representation) represents a substring, which is not
in $\cF$. When the last vertex in the path is generated by the current prefix $\bfs'$ of the sequence $\bfs$ is vertex $v$,
it implies that the substring represented by vertex $v$ is the longest suffix of $\bfs'$ which is a prefix
of a forbidden pattern in $\cF$. Each one of the $\sigma$ out-edges
exist in $v$ if and only each one does not lead to a forbidden substring in $\cF$.
Eventually, both representations are equivalent and one can be derived from the other
by using a reduction of the state diagram~\cite{MRS01}.
But, each representation has some advantages on the other one.
The advantage of using the state diagram of the constraint, based on the forbidden
substrings, is a simple definition for the state diagram from which the
theory of constrained codes can be used. The advantage of using a graph like a de Bruijn graph is that
each point on the sequence is related to a specific $k$-tuple and it is identified by a vertex (note that some $k$-tuples are
represented by a few vertices) which represents this $k$-tuple. An example of the state diagram for the $(3,3)$-constrained
de Bruijn sequences via the forbidden patterns is depicted in Figure~\ref{Fig2}.

\begin{figure}
\centering
\begin{tikzpicture}[->,>=stealth',shorten >=1pt,auto,node distance=2.5cm,
                    thin,main node/.style={circle,draw,font=\sffamily\Large\bfseries}]

  \node[main node] (1) {x};
  \node[main node] (2) [ left of=1] {0};
  \node[main node] (3) [ right of=1] {1};
  \node[main node] (4) [below left of=2] {00};
  \node[main node] (5) [below right of=2] {01};
    \node[main node] (7) [below right of=3] {11};
  \node[main node] (6) [below left of=3] {10};

   \node[main node] (8) [ below left of=4] {000};
     \node[main node] (9) [ below left of=5] {010};
      \node[main node] (10) [ below right of=6] {101};
      \node[main node] (11) [ below right of=7] {111};
      \node[main node] (12) [ below of=10] {0101};
 \node[main node] (13) [ below of=9] {1010};

  \path[every node/.style={font=\sffamily\small}]
       (1) edge node [above] {$0$} (2)
       (1) edge node [above] {$1$} (3)
       (2) edge node[left] {$0$} (4)
       (2) edge node[right] {$1$} (5)
       (3) edge node [left] {$0$} (6)
       (3) edge node [right] {$1$} (7)
       (4) edge node [right] {$0$} (8)
       (4) edge node [above] {$1$} (5)
		(5) edge node [right] {$0$} (9)
		(5) edge [bend right] node [above] {$1$} (7)
		(6) edge [bend left] node [above] {$0$} (4)
		(6) edge node [right] {$1$} (10)
		(7) edge node [above] {$0$} (6)
		(7) edge node [right] {$1$} (11)
		(8) edge node [left] {$1$} (5)
		(9) edge node [above] {$1$} (12)
		(9) edge node [above] {$0$} (4)
		(11) edge node [left] {$0$} (6)
		(10) edge node [below] {$0$} (13)
		(10) edge node [above] {$1$} (7)
		(12) edge node [right] {$1$} (7)
		(13) edge node [left] {$0$} (4)

    ;
\end{tikzpicture}
\caption{State diagram for the $(3,3)$-constrained de Bruijn sequences, via the forbidden patterns\vspace{-3ex}}
\label{Fig2}
\end{figure}

\subsection{State Diagrams and Rates based on the Forbidden Subsequences}
\label{sec:state_constrained}

As a sequence of a constrained code, the $(b,k)$-constrained de Bruijn sequence has
several forbidden patterns, $\bfs_1,\bfs_2,\ldots,\bfs_\ell$. W.l.o.g. we assume that $\bfs_i$ is
not a prefix of $\bfs_j$ for $i \neq j$. The state diagram has an initial vertex $x$,
from which there are $\ell$ paths. The $i$th path has length smaller by one than the length
of the $i$th forbidden pattern $\bfs_i$. Paths share a prefix if their related sequences
share a corresponding prefix. Each vertex in the state diagram, except for the ones at the
end of the $\ell$ paths, has out-degree $\sigma$, one edge for each possible symbol of $\Sigma$ that can be seen
in the corresponding read point of the constrained sequence. The edges are labelled by the related symbols.
The last symbol of $\bfs_i$ does not appear of an edge from the last vertex of the related path.
A vertex in the state diagram is labelled by the prefix of the path which it represents. The initial
vertex $x$ is labelled with the sequence of length zero. Thus, to construct the state diagram
we have to determine all the forbidden substrings in the $(b,k)$-constrained de Bruijn sequence.
This representation coincides with the representation of the constrained system and the related
example for the $(3,3)$-constrained de Bruijn sequences is depicted in Figure~\ref{Fig2}.

To determine exactly the maximum asymptotic rates of $(b,k)$-constrained de Bruijn codes,
we use the well-known Perron-Frobenius theory~\cite{MRS01}. When $b$ and $k$ are given,
by using the the de Bruijn graph for the given constraint or the related state diagram,
we can build a finite directed graph with labelled edges such that paths in the graph
generate exactly all $(b,k)$-constrained de Bruijn sequences.
For example, when $(b,k)=(3,3)$ the
adjacency matrix is
\begin{small}
\[
A_{G}=
\begin{pmatrix}
0 & 1 & 0 & 0& 0& 0& 0& 0& 0& 0\\
0 & 0 & 0 & 1& 0& 0& 0& 1& 0& 0\\
1 & 1 & 0 & 0& 0& 0& 0& 0& 0& 0\\
0 & 0 & 1 & 0& 0& 1& 0& 0& 0& 0\\
0 & 0 & 1 & 0& 0& 0& 0& 0& 0& 0\\
0 & 0 & 0 & 0& 0& 0& 0& 1& 0& 0\\
0 & 0 & 0 & 0& 1& 0& 0& 1& 0& 0\\
0 & 0 & 0 & 0& 0& 0& 0& 0& 1& 1\\
0 & 0 & 1 & 0& 0& 0& 1& 0& 0& 0\\
0 & 0 & 0 & 0& 0& 0& 0& 0& 1& 0\\
\end{pmatrix}
\]
\end{small}
Its largest eigenvalue is $\lambda \approx 1.73459$. Hence, the capacity
of this constrained system, 
which is the maximum asymptotic rate of $(3,3)$-constrained de Bruijn code, is $\log \lambda = 0.7946$.
Similarly, we can compute the maximum asymptotic rates of $(b,k)$-constrained de Bruijn codes for
other values of $b,k$.
Table~\ref{table1} presents some values for the asymptotic rates of the constrained systems for
small parameters.

\begin{table}[h]
\caption{The maximal asymptotic rates of $(b,k)$-constrained de Bruijn codes.}
\begin{center}
\begin{tabular}{|c|c|c|c|c|c|c|c|c|c| }
 \hline
     & $k=2$ & $k=3$ & $h=4$ & $k=5$ & $k=6$ & $k=7$ &$k=8$&$k=9$&$k=10$\\
     \hline
 $b=2$ & 0.6942 & 0.8791 & 0.9468 & 0.9752 & 0.9881 & 0.9942& 0.9971&0.9986&0.9993\\
 $b=3$ & 0.4056 & 0.7946 & 0.9146 & 0.9614 & 0.9817&0.9912&0.9957&0.9978&0.9989 \\
 $b=4$ & 0	      & 0.6341  &	 0.8600	& 0.9392	& 0.9719&0.9865&0.9934&0.9966&0.9978\\
 $b=5$ & 0	      & 0.4709  &	 0.7973	& 0.9150& 0.9615&0.9818&0.9912&0.9957&0.9978\\
  $b=6$ & 0	      & 0.4517  &	 0.7289	& 0.88412& 0.94815&0.97574&?&?&?\\
 \hline
\end{tabular}
\end{center}
\label{table1}
\end{table}

The asymptotic rate can be evaluated for infinite pairs of $(b,k)$ as proved in the next theorem.

\begin{theorem}
\label{thm:rate_explicit}
For any positive integer $k>1$, the maximum asymptotic rate of a binary
$(3,k)$-constrained de Bruijn code is $\log_2 \lambda$, where $\lambda$ is
the largest root of the polynomial
$x^{2k-1}=x^{2k-3}+2x^{2k-4}+\cdots+(k-2)x^k+(k-1)x^{k-1}+(k-1)x^{k-2}+(k-2)x^{k-3}+\cdots+2x+1.$
\end{theorem}
\begin{proof}
Recall that $\cF_{p,p+k}$ is the set of all period $p$ sequences of length $p+k$.
Let $\cF_1=\cF_{1,k+1} \cup \cF_{2,k+2}$, i.e., $\cF_1$ contains
the all-zero word of length $k+1$, the all-one word of length $k+1$, and
the two words of length $k+2$ in which any two consecutive positions
has distinct symbols. By Corollary~\ref{cor:relation}, $\C_{DB}(n,3,k)=\cA (n;\cF_1)$, i.e.
binary $(3,k)$-constrained de Bruijn code of length $n$ is an $\cF_1$-avoiding code of length $n$.

Let $\cF_2$ by the set which contains the all-zero word of length $k$ and the all-one word of length $k+1$.
Consider the ${\bf D}$-morphism defined first in~\cite{Lem70},
${\bf D}: B^n \mapsto B^{n-1}$, $B=\{ 0,1 \}$, where
${\bf D} (\bfx)={\bf D}(x_1,x_2,\ldots,x_n) =\bfy=(y_2,\ldots,y_n)$, with $y_i=x_i + x_{i-1}$, $2 \leq i \leq n$.
It was proved in~\cite{Lem70} that the mapping ${\bf D}$ is a 2 to 1 mapping.
Furthermore, $\bfx \in \cA(n;\cF_1)$ if and only if $\bfy \in \cA(n-1;\cF_2)$.
Hence, $|\cA(n;\cF_1)| = 2|\cA(n-1;\cF_2)|$. This implies that
$$
\lim_{n \to \infty } \frac{\log_2 |\cA(n;\cF_2)|}{n}=\lim_{n \to \infty } \frac{ \log_2 |\cA(n;\cF_1)|}{n}.
$$
Hence, $\cA(n;\cF_1)$ and $\cA(n;\cF_2)$ have the same maximum asymptotic rate, and hence
$\cA(n;\cF_2)$ can be computed instead of $\cA(n;\cF_1)$.

Let $\cA_0(n;\cF_2)$ be the set of all $\cF_2$-avoiding words of
length $n$ which start with a \emph{zero} and let $\cA_1(n;\cF_2)$ be the set of all $\cF_2$-avoiding words
of length $n$ which start with a \emph{one}. Clearly, the asymptotic rates
of $\cA(n;\cF_2)$, $\cA_0(n;\cF_2)$, $\cA_1(n;\cF_2)$ are equal.
\noindent
Let
$$
\Phi_1: \cA_0(n;\cF_2) \mapsto \cup_{i=1}^{h-1} \cA_1(n-i;\cF_2)
$$
be the mapping for which
$\Phi_1(\bfx = (0,x_2,\ldots,x_n))=(x_{i+1},\ldots,x_n) \in  \cA_1(n-i;\cF_2)$,
where $i$ is the smallest index that $x_{i+1}=1$. Since $\cA_0(n;\cF_2)$ avoids the all-zero sequence of length $k$,
it follows that the mapping $\Phi_1$ is a well-defined bijection.
Therefore,
\begin{equation}
\label{eq1}
|\cA_0(n;\cF_2)| = \sum_{i=1}^{k-1} |\cA_1(n-i;\cF_2)|.
\end{equation}
Similarly, we can define the bijection
$$
\Phi_2: \cA_1(n;\cF_2) \mapsto \cup_{i=1}^{k} \cA_0(n-i;\cF_2)$$
and obtain the equality
\begin{equation}
\label{eq2}
|\cA_1(n;\cF_2)| = \sum_{i=1}^{k} |\cA_0(n-i;\cF_2)|.
\end{equation}
Equations (\ref{eq1}) and (\ref{eq2}) imply that
\begin{align*}
&|\cA_0(n;\cF_2)| = \sum_{i=1}^{k-1} |\cA_1(n-i;\cF_2)|= \sum_{i=1}^{k-1} \sum_{j=1}^{k} |\cA_0(n-i-j;\cF_2)|\\
&=\sum_{\ell=2}^{k} (\ell-1)|\cA_0(n-\ell;\cF_2)| + \sum_{\ell=1}^{k-1} \ell |\cA_0(n-2h+\ell;\cF_2)| \\
&= |\cA_0(n-2;\cF_2)| + 2|\cA_0(n-3;\cF_2)|+ \cdots + (k-1)|\cA_0(n-k;\cF_2)|\\
&+(k-1)|\cA_0(n-k-1;\cF_2)|+ (k-2)|\cA_0(n-k-2;\cF_2)| +\cdots+|\A_0(n-2k+1;\cF_2)|.
\end{align*}

It is again easy to verify that the maximum asymptotic rates of $\cA_0(n-\ell;\cF_2)$ for
all $2 \leq \ell \leq 2k-1$ are equal. Let $\lambda^n$ be this maximum asymptotic rate.
The recursive formula can be solved now for $\lambda$ and the maximum asymptotic rate
of $\cA_0(n;\cF_2)$ will be $\log_2 \lambda$, where~$\lambda$ is computed as the largest root of the polynomial $x^{2k-1}=x^{2k-3}+2x^{2k-4}+\cdots+(k-2)x^k+(k-1)x^{k-1}+(k-1)x^{k-2}+(k-2)x^{k-3}+\cdots+2x+1.$
\end{proof}

Using recursive formulas in Equations~(\ref{eq1}) and~(\ref{eq2}), we can compute the
exact size of $\cA(n;\cF_2)$ efficiently. Hence, we can rank/unrank all words
in $\A(n;\cF_2)$ efficiently using enumerative technique~\cite{Cover73}.

Computing the largest eigenvalue of the adjacency matrix is the key for computing the asymptotic
rate. The size of the matrix and the number of its nonzero entries is important for
reducing the complexity of the computation. Fortunately we can evaluate some of these parameters
to get some idea for which parameters it is feasible to compute the largest eigenvalue. For lack of space and since this is
mainly a combinatorial problem we omit this computation and leave it for a future work.
Finally, we note that we can use the idea in the proof of Theorem~\ref{thm:rate_explicit} to reduce the
number of forbidden patterns by half, while keeping the same asymptotic rate. This is done
in the following example.

\begin{example}
Consider the $(3,3)$-constrained de Bruijn sequences. The set of patterns which should be avoided is $\{0000,1111,01010,10101\}$.
As was illustrated in Figure~\ref{Fig1} and Figure~\ref{Fig2}, 10 nodes are required to represent
a graph of constrained sequences avoiding these four patterns. Using Theorem 9,
it can be observed that the asymptotic size of this code
is the same as the code avoiding all patterns in set $\{000,1111\}$. Hence, only 5 nodes are required to represent
the state diagram of the constrained sequences avoiding these two patterns as depicted in Figure \ref{Fig3}.

\begin{figure}
\centering
\begin{tikzpicture}[->,>=stealth',shorten >=1pt,auto,node distance=2.5cm,
                    thin,main node/.style={circle,draw,font=\sffamily\Large\bfseries}]

  \node[main node] (1) {x};
  \node[main node] (2) [below left of=1] {0};
  \node[main node] (3) [below right of=1] {1};
  \node[main node] (4) [below  of=2] {00};
    \node[main node] (5) [below of=3] {11};
  \node[main node] (6) [below left of=5] {111};

  \path[every node/.style={font=\sffamily\small}]
       (1) edge node [above] {$0$} (2)
       (1) edge node [above] {$1$} (3)
       (2) edge [bend right] node[above] {$1$} (3)
       (2) edge node[left] {$0$} (4)
       (3) edge  node [above] {$0$} (2)
       (3) edge  node [right] {$1$} (5)
       (4) edge node [right] {$1$} (3)
		(5) edge node [right] {$1$} (6)
		(5) edge [bend left] node [below] {$0$} (2)
		(6) edge node [below] {$0$} (2)
    ;
\end{tikzpicture}
\caption{Reduced state diagram to count the number of $(3,3)$-constrained de Bruijn sequences\vspace{-3ex}}
\label{Fig3}
\end{figure}
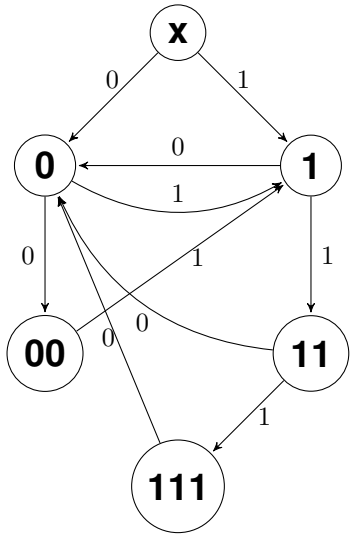

\end{example}

\section{Codes with a Large Constrained Segment}
\label{sec:construct}

After considering in Section~\ref{sec:enumerate} enumerations and constructions
for $(b,k)$-constrained de Bruijn sequences which are mainly efficient
for small $b$ compared to $k$, we will use the theory of shift registers for a construction
which is applied for considerably larger $b$ compared to $k$. Such constructions for
$(\sigma^k,k)$-constrained de Bruijn codes are relatively simple as each sequence in the code is formed by
concatenations of the same de Bruijn sequence. The disadvantage is
that the rate of the related code is zero. In this section we present a construction of codes in which $b$
is relatively large compare to $k$ and the rate of the code is greater than zero.

Let $\cP_k$ be the set of all primitive polynomial of degree $k$ over $\F_q$. By Theorem~\ref{thm:num_primitive},
there are $\frac{\phi (q^k-1)}{k}$ polynomials in $\cP_k$. Each polynomial is
associated with a linear feedback shift register and a related m-sequence of length $q^k-1$.
Let $\cS_{\cP_k}$ be this set of m-sequences.
In each such sequence, each nonzero $k$-tuple over $\F_q$
appears exactly once as a window in one period of the sequence. We can single out that
in each such sequence there is a unique  run of $k-1$ \emph{zeroes} and a unique run of length $k$ for each other symbol.
This run of the same symbol is the longest such run in each sequence.
But, there is another window property for all the sequences in~$\cS_{\cP_k}$.

\begin{lemma}
Each $(2k)$-tuple over $\F_q$ appears at most once as a window in one of the sequences of~$\cS_{\cP_k}$.
\end{lemma}
\begin{proof}
Let $f_1(x)$ and $f_2(x)$ be two distinct primitive polynomials in $\cP_k$, whose state diagrams
contain the cycles of length $q^k-1$, $\cC_1$ and $\cC_2$, respectively. By Theorem~\ref{thm:multiply_two_primitive},
the polynomial $f_1(x) f_2(x)$
has degree $2k$ and its state diagram contains the two cycles $\cC_1$ and $\cC_2$.
Each $(2k)$-tuple appears exactly once in a window of length~$2k$ in one
of the cycles of the state diagram related to $f_1(x) f_2(x)$. Hence, each such
$(2k)$-tuple appears at most once in either $\cC_1$ or $\cC_2$.

The windows of length $k$ are distinct within each sequence and the windows of length $2k$
are distinct between any two sequences which implies the claim of the lemma.
\end{proof}

We are now in a position to describe our construction of constrained de Bruijn codes with
a large constrained segment.

\begin{construction}
\label{con:basic}
Let $\cS_{\cP_k}$ be the set of all m-sequences of order $k$ over $\F_q$. Each sequence will be considered
as an acyclic sequence of length $q^k-1$, where the unique run of $k-1$ \emph{zeroes} is at the end of the sequence.
We construct the following code

$$
\C \triangleq \{ ( 0^{\epsilon_0=0} \bfs_1 0^{\epsilon_1} ,\bfs_2 0^{\epsilon_2},\ldots,
 \bfs_{\ell-1} 0^{\epsilon_{\ell-1}}  ,\bfs_\ell 0^{\epsilon_\ell}) ~:~
           \bfs_i \in \cS_{\cP_k}, ~ 0 \leq \epsilon_i \leq k+1,~ 1 \leq i \leq \ell \}~.
$$
\end{construction}

\begin{theorem}
\label{thm:basic}
The code $\C$ contain $(q^k-1,2k)$-constrained de Bruijn sequences, each one of length
at least $\ell (q^k -1)$ and at most $\ell (q^k +k)$. The
size of $\C$ is $M^\ell$, where $M=\frac{\phi(q^k-1)(k+2)}{k}$ and $k \geq 3$.
\end{theorem}
\begin{proof}
We first prove that each codeword of $\C$ is a $(q^k-1,2k)$-constrained de Bruijn sequence.
Let $\bfs=0^{\epsilon_{i-1}} \bfs_i 0^{\epsilon_i}$, $ 1 \leq i \leq \ell$, be a substring of a codeword in $\C$.
We first note that except for the windows of length $2k$ having more than $k-1$ \emph{zeroes} at the start or at
the end of $\bfs$, all the other windows of length $2k$ contained in $\bfs$ appear in the cyclic sequence $\bfs_i$.
These new windows (having more than $k-1$ \emph{zeroes}) as well as the windows which contain the run of \emph{zeroes} $0^{\epsilon_i+k-1}$
appear only between two sequences $\bfs_j$ and $\bfs_{j+1}$ and as such they are separated by at least
$q^k-k$ symbols. Moreover, each window of length $2k$ containing \emph{zeroes} from one such
run is clearly unique. This implies that each sequence of $\C$ is a $(q^k-1,2k)$-constrained de Bruijn sequence
(as repeated windows of length $2k$ can occur only between different $s_i$'s).
Since each $\bfs_i$ has length $q^k-1$ and $0 \leq \epsilon_i \leq k+1$, it follows that the length of a codeword
is at least $\ell (q^k -1)$ and at most $\ell (q^k +k)$.
The number of sequences that can be used for the substring $\bfs_i 0^{\epsilon_i}$ is $\frac{\phi(q^k-1)(k+2)}{k}$
since $\bfs_i$ can be chosen in $\frac{\phi(q^k-1)}{k}$ ways (see Theorem~\ref{thm:num_primitive}) and
$0^{\epsilon_i}$ can be chosen in $k+2$ distinct ways. It implies that $|\C| = M^\ell$.
\end{proof}

The codewords in the code $\C$ obtained via Construction~\ref{con:basic} can be of different lengths.
We will now construct a similar code in which all codewords have the same length.
Let $\C'$ be a code which contains all the prefixes of codewords from $\C$.
Let
$$
\C_1 \deff \{ (\bfs_1 \bfs_2) ~:~ \bfs_1 \in \C , ~ \bfs_2 \in \C',~ \text{length}(\bfs_1 \bfs_2) =\ell (q^k +k) \}
$$

\begin{theorem}
\label{thm:basic1}
The code $\C_1$ contain $(q^k-1,2k)$-constrained de Bruijn sequences, of length ${n=\ell (q^k +k)}$ and its
size is at least $M^\ell$, where $M=\frac{\phi(q^k-1)(k+2)}{k}$.
\end{theorem}
\begin{proof}
The codewords in $\C_1$ are formed from the codewords in $\C$ by lengthening them to the required length
with prefixes of codewords in $\C$. The lengthening does not change the structure of the codewords, it
just change their length. Hence, the proof that the codewords are $(q^k-1,2k)$-constrained de Bruijn sequence
is the same as in the proof of Theorem~\ref{thm:basic}. The number of codewords is the same as in $\C$ if there
is exactly one way to lengthen each codeword of $\C$. Since there is usually more than
one such way, it follows that the code will be of a larger size.
\end{proof}
\begin{corollary}
The rate of the code $\C_1$ is $\frac{k}{q^k+k}$.
\end{corollary}

Construction~\ref{con:basic} yields acyclic sequences. Can we find a related construction
for cyclic constrained de Bruijn sequences with similar parameters? The answer is definitely
positive. Construction~\ref{con:basic} can be viewed as a construction for cyclic sequences of different length.
To have a code with cyclic $(q^k-1,2k)$-constrained de Bruijn sequences of the same length (as the acyclic sequences in $\C_1$,
we can restrict the values in the $\epsilon_i$'s. For example, we can
require that $\sum_{i=1}^{\ell} \epsilon_i = \lfloor \frac{\ell (k+1)}{2} \rfloor$.
This will imply similar results for the cyclic code as for the acyclic code.

The code $\C_1$ can be slightly improve by size, but as it does not make a significant
increase in the rate we ignore the possible improvements. The same construction can be applied
with a larger number of cycles whose length is $\sigma^k-1$ ($\sigma$ not
necessarily a prime power) and with a the same window length or even
a smaller one. We do not have a general construction for such cycles, but we outline the simple
method to search for them.

Let $\cD(\sigma,k)$ be the set of de Bruijn sequences in $G_{\sigma,k}$. Let $\delta$ be a positive integer.
Construct a graph whose vertices are all the de Bruijn sequences in $\cD(\sigma,k)$.
Two vertices (de Bruijn sequences) are connected by an edge if they have the same window
of length greater than $k+\delta$. Now let $\cT$ be an independent set in the graph.
Two sequences related to this independent set do not share a window of length $k+\delta$
or less, so we can apply Construction 1 and its variants, where $0 \leq \epsilon_i \leq \delta$.
Note that in Construction 1, the set of sequences in $\cS_{\cP_k}$ form an independent set with $\delta=k$.

A computer search for the 2048 binary de Bruijn sequences of length 32 was performed.
An independent set of size 8 was found for $\delta =5$ compared to the 6 m-sequences of length 31.
Furthermore, for $\delta=4$, an independent set of size 4 was found, and for $\delta=2$ the size of
the independent set that was found was only 2. This implies that this direction of research can be quite promising.

\section{Application to the $\ell$-Symbol Read Channel}
\label{sec:symbol_read}

In this section, we show that cyclic $(b,k)$-constrained de Bruijn codes can be used
to correct synchronization errors in the $\ell$-symbol read channel~\cite{CB11,CJKWY13,YBS16}.
Previously, only substitution errors were considered in such a channel.
Each cyclic $(b,k)$-constrained de Bruijn sequence which forms a codeword in the channel
can be used to correct a number of limited synchronization errors which might occurred.
The correction does not depend on the other codewords of the code.
The mechanism used in this section will be of help in correcting such errors in  racetrack memory as
discussed in the next section. We will consider now~$\F_q$ as our alphabet although the method
can be applied on alphabets of any size.
The reason is that other error-correcting codes, for this purpose, are defined over $\F_q$.

\begin{definition}
Let $\bfx=(x_1,x_2,\ldots,x_n) \in \F_q^n$ be a $q$-ary sequence of length $n$. In the \emph{$\ell$-symbol read channel},
if $\bfx$ is a codeword then the corresponding $\ell$-symbol read vector of $\bfx$ is
$$
\pi_{\ell}(\bfx)=((x_1,\ldots,x_{\ell}), (x_2,\ldots,x_{\ell+1}),\ldots,(x_{n},x_1\ldots,x_{\ell-1})).
$$
\end{definition}

Note first, that $\bfx$ might not be a cyclic sequence, but the channel read the symbols cyclically, i.e.,
after the last symbol, the first symbol, second symbols and so on are read to complete an $\ell$-symbol read
in each position of $\pi_{\ell}(\bfx)$.
In this channel, the $\ell$-symbol read sequence $\pi_{\ell}(\bfx)$ is received with no error,
it is simple to recover the codeword $\bfx$. However, several types of errors,
such as substitution errors and synchronization errors, can occur.
Substitution errors were considered in~\cite{CB11,CJKWY13,YBS16}, but a synchronization error
might be a more common one.
We focus now on the $\ell$-symbol read channel with only synchronization errors which are deletions
and sticky insertions. A \emph{sticky insertion} is the error event when an $\ell\text{-tuple}$
$(x_i,\ldots,x_{i+\ell-1})$ is repeated, i.e., read more than once. A deletion is the error event when an $\ell$-tuple
is deleted. A burst of deletions of length at most~$b$ is an error
in which at most~$b$ consecutive $\ell$-tuples are deleted.
\begin{example}
Let $\bfx=(0,1,0,0,1,0,0,0)$ be a codeword. When $\ell=2$, the 2-symbol read sequence of~$\bfx$ is $\pi_2(\bfx)=((0,1),{\color{red}(1,0)},(0,0),(0,1),(1,0),{\color{blue}(0,0),(0,0)},(0,0))$.
If there is a sticky insertion at the second location and a burst of deletions
of length two at the 6-th and 7-th positions, we obtain the 2-symbol read sequence
$\pi_2(\bfy)=((0,1),{\color{red}(1,0),(1,0)},(0,0),(0,1),(1,0),(0,0))$.
\end{example}
The goal of this section is to show that a cyclic constrianed de Bruijn code is
a code correcting sticky insertions and bursts of deletions in the $\ell$-symbol read channel.
\begin{theorem}
\label{thm:cyclic_l_symbol}
Let $b$ and $k$ be two positive integers such that ${q^k \geq b\geq 2}$ and $\ell=k+b-2$.  The code $\C'_{DB}(n,b,k)$
can correct any number of sticky-insertions and any number of bursts of deletions (depending
only on $n$ and $b$) whose length is at most $b-2$
(including one possible cyclic burst of this size) in the $\ell$-symbol read channel.
\end{theorem}
\begin{proof}
Let $\bfc =(c_0,c_1,\ldots,c_{n-1})$ be a codeword and
$$
\pi_{\ell}(\bfc)= (((c_0,\ldots,c_{\ell-1}),(c_1,\ldots,c_{\ell}),\ldots,(c_{n-1},\ldots,c_{n+\ell-2})),
$$
where indices are taken modulo $n$, be a corresponding $\ell$-symbol read sequence of $\bfc$.
Let
$$
\pi_{\ell}(\bfy) = ((y_{1,1},\ldots,y_{1,\ell}),(y_{2,1},\ldots,y_{2,\ell}),\ldots,(y_{t,1},\ldots,y_{t,\ell}))
$$
be a received $\ell$-symbol read sequence.

Clearly, since no more than $b-2$ consecutive deletions occurred, it follows that there exists an $i$, $1 \leq i \leq b-1$,
such that $(c_i , c_{i+1},\ldots,c_{i+\ell-1})=(y_{1,1},\ldots,y_{1,\ell})$. Similarly, there also exists a $j$, $0 \leq j \leq b-1$,
such that $(c_{i+j} , c_{i+j+1},\ldots,c_{i+j+\ell-1})=(y_{2,1},\ldots,y_{2,\ell})$.

Since $\ell=k+b-2$ and at most $b-2$ deletions occurred between the first two consecutive $\ell$-reads, it follows that
either the substring $(y_{2,1},\ldots,y_{2,k})$ is either a substring of $(y_{1,1},\ldots,y_{1,\ell})$ or
$(y_{1,b},y_{1,b+1},\ldots,y_{1,\ell})=(y_{2,1},\ldots,y_{2,k-1})$.
If $(y_{2,1},\ldots,y_{2,k})$ is a substring of $(y_{1,1},\ldots,y_{1,\ell})$, then since any $b$ consecutive {$k\text{-tuples}$}
of $\bfc$ are distinct, it follows that there exists a unique $j$, $0 \leq j \leq b-2$, such that
$(c_{i+j} , c_{i+j+1},\ldots,c_{i+j+k-1})=(y_{2,1},\ldots,y_{2,k})$. If $j=0$ then a sticky insertion
has occurred and if $j=1$ no error occurred in the second read, and if $2 \leq j \leq b-2$, then a burst of
$j-1$ deletions has occurred. If a burst of $b-2$ deletion has occurred, then $(y_{1,b},y_{1,b+1},\ldots,y_{1,\ell})=(y_{2,1},\ldots,y_{2,k-1})$.
The same process continue between the second read and the third read and so on until the last read and the first read
which behave in the same way since also cyclically the longest run of deletions has at most length $b-2$.
Thus, using the overlaps between consecutive reads we can recover the codeword $\bfc$.
\end{proof}

\section{Applications to Racetrack Memories}
\label{sec:racetrack}

Racetrack memory is an emerging non-volatile memory technology which has attracted significant attention
in recent years due to its promising ultra-high storage density and low power consumption\cite{Parkinetal08, Sunetal13}.
The basic information storage element of a racetrack memory is called a \emph{domain}, also known as a \emph{cell}.
The magnetization direction of each cell is programmed to store information. The reading mechanism
is operated by many \emph{read ports}, called \emph{heads}. In order to read the information,
each cell is shifted to its closest head by a \emph{shift operation}. Once a cell is shifted,
all other cells are also shifted in the same direction and in the same speed.
Normally, along the racetrack strip, all heads are fixed and distributed uniformly\cite{Zhangetal15}.
Each head thus reads only a block of consecutive cells which is called a \emph{data segment}.

A shift operation might not work perfectly. When the cells are not shifted (or under shifted), the same cell
is read again in the same head. This event causes a repetition (or sticky insertion) error.
When the cells are shifted by more than a single cell location (or over shifted), one cell or
a block of cells is not read in each head. This event causes a single deletion or a burst of
consecutive deletions. We note that the maximum number of consecutive deletions is limited or in other words,
the burst of consecutive deletions has limited length. An experimental result shows that the cells
are over shifted by at most two locations with extremely high probability\cite{Zhangetal15}.
In this paper, we study both kinds of errors and refer to these errors as \emph{limited-shift errors}.

Since limited-shift errors can be modeled as sticky insertions and bursts of consecutive deletions with limited length,
sticky-insertion/deletion-correcting codes can be applied to combat these limited-shift errors. Although there are several known
\emph{sticky-insertion-correcting codes}~\cite{DA10,JFSB16,MV17},
\emph{deletion-correcting codes}~\cite{BGZ17,Hel02,LV66}, \emph{single-burst-deletion-correcting codes}~\cite{Chengetal14,SWGY16},
and \emph{multiple-burst-deletion-correcting codes}~\cite{HR17}, there is a lack of knowledge on codes correcting
a combination of multiple bursts of deletions and sticky insertions. Correcting these type of errors
are especially important in the racetrack memories. In this section, motivated by the special structure
of having multiple heads in racetrack memories, we study codes correcting multiple bursts of deletions and sticky insertions.
To correct shift errors in racetrack memories with only a single head, Vahid et al.~\cite{VMSC17} recently studied
codes correcting two shift errors of deletions and/or insertions.

Another approach to combat limited-shift errors is to leverage the special feature of racetrack memories
where it is possible to add some extra heads to read cells. If there is no error, the information read in
these extra heads is redundant. However, if there are limited-shift errors, this information is useful
to correct these errors. Recently, several schemes have been proposed to leverage this
feature~\cite{CKVVY17A,CKVVY18,Zhangetal15} in order to tackle this problem.
However, in~\cite{CKVVY17A,CKVVY18}, each head needs to read all the cells while in this model,
each head only needs to read a single data segment. Our goal in this section is to present several schemes
to correct synchronization errors in racetrack memories, all of them are based on
constrained de Bruijn sequences and codes.
In some of these schemes we add extra head and some
are without adding extra heads.

\subsection{Correcting Errors in Racetrack Memories without Extra Heads}
\label{sec:no_more_heads}

Our first goal in this subsection is to construct $q$-ary
$b$-limited $t_1$-burst-deletion-correcting codes
to combat synchronization errors in racetrack memories.
Such a code can correct $t_1$ bursts of deletions if the length
of each such burst is at most $b$, i.e., at most $b$ deletions occurred
in each such burst and each two of these bursts are separated by symbols
which are not in any error.

\begin{lemma}
\label{lem:construct_blkofdels}
Let $\bfs$ be a $(b,1)$-constrained de Bruijn sequence over an alphabet of size $q$.
Let $\bfs(\Delta^-)$ be the sequence obtained from $\bfs$ after deleting all symbols specified
by the locations in the set $\Delta^-$ such that the number of consecutive deletions is at most $b-1$.
Then the set $\Delta^-$ is uniquely determined from $\bfs$ and $\bfs(\Delta^-)$.
\end{lemma}
\begin{proof}
Assume that $\Delta^-=\{i_1,i_2,\ldots,i_t\}$ where $i_1 < i_2 < \cdots < i_t$ is the set of $t$ locations of all $t$ deleted symbols.
Since $i_1$ is the leftmost index in which a deletion has occurred,
it follows that $\bfs[1,i_1-1]=\bfs(\Delta^-)[1,i_1-1]$. Furthermore, since $\bfs$ is
a $(b,1)$-constrained de Bruijn sequence, it follows that the symbols $\bfs[i]$ for $i_1 \leq i \leq i_1 +b-1$ are distinct.
Since $\bfs[i_1]$ was deleted and the maximum number of consecutive deletions is at most $b-1$, it follows that
$\bfs(\Delta^-)[i_1] \neq \bfc[i_1]$. So, $i_1$ is the leftmost index in which $\bfs$ and $\bfs(\Delta^-)$ differ.
Therefore, we can determine $i_1$ from the two vectors $\bfc$ and $\bfs(\Delta^-)$.
Let $\Delta_1^- \triangleq \Delta^- \setminus \{i_1\} =\{i_2,\ldots,i_t\}$.
To correct the first error, we insert the symbol $\bfs[i_1]$ into the $i_1$-th position of $\bfs(\Delta^-)$ and obtain the vector $\bfs(\Delta^-_1)$. Similarly, we can continue now to determine, one by one, the other positions where deletions have occurred using
words $\bfs$ and $\bfs(\Delta_1^-)$.

Thus, the set $\Delta^-$ can be determined from $\bfc$ and $\bfc(\Delta^-)$ and the lemma is proved.
\end{proof}

We are now ready to present a construction of $q$-ary
$b$-limited $t_1$-burst-deletion-correcting codes.
We will show that the maximum rate of these codes is close to the maximum rate of codes correcting multiple erasures,
especially when $q$ is large. A $q$-ary \emph{$t$-erasure-correcting code of length $\ell$} is a set of
$q$-ary words of length $\ell$ in which one can correct any set of $t$-erasures, i.e. $t$ known positions whose values are not known.
\begin{construction}
\label{con:blkofdels}
Let $\bfs=(s_1,s_2,\ldots,s_{\ell})$ be a $(b,1)$-constrained de Bruijn sequence
over an alphabet of size $q_1$. Let $\C_{q_2}(\ell,t)$ be a $q_2$-ary
$t$-erasure-correcting code of length $\ell$ and
let $q=q_1 \cdot q_2$. For each word $\bfc=(c_1,c_2,\ldots,c_{\ell}) \in \C_{q_2}(\ell,t)$,
we define $\bff(\bfc,\bfs)=(f_1,f_2,\ldots,f_{\ell})$, where $f_i=(c_i,s_i)$ for all $1\leq i \leq \ell$.
Construct the $q$-ary code of length $\ell$ which is denoted by $\C_q(b,\ell,t)$
$$
\C_q(b,\ell,t) \triangleq \{\bff(\bfc,\bfs): \bfc \in \C_{q_2}(\ell,t) \}~.
$$
\end{construction}

\begin{theorem}
\label{thm:construct_blkofdels}
The code $\C_q(b,\ell,t)$ obtained in Construction~\ref{con:blkofdels} is a
$q$-ary $(b-1)$-limited $t_1$-burst-deletion-correcting code where $t= t_1\cdot (b-1)$.
\end{theorem}

\begin{proof}
Let $\bff=(f_1,\ldots,f_{\ell}) \in \C_q(b,\ell,t)$ be a codeword of length $\ell$.
Let $\Delta^-=\{i_1,\ldots,i_t\}$ be the set of all deleted positions
such that $i_1 < \cdots < i_t$ and let $\bff(\Delta^-)$ denote the received word.
Since there are at most $t_1$ bursts of deletions whose length is at most $b-1$, it follows that
$|\Delta^-|=t \leq t_1 \cdot (b-1)$ and there are at most $b-1$ consecutive integers in $\Delta^-$.
From the received word $\bff(\Delta^-)$, we can extract the unique pair of words $(\bfc(\Delta^-),\bfs(\Delta^-))$,
where $\bfc(\Delta^-)$  and $\bfs(\Delta^-)$ are two words obtained from $\bfc$ and $\bfs$, respectively, after
deleting all the symbols specified by the locations in the set $\Delta^-$. By Lemma~\ref{lem:construct_blkofdels},
the set $\Delta^-$ can be determined from the words $\bfs(\Delta^-)$ and $\bfs$ since $\bfs$ is
a $(b,1)$-constrained de Bruijn sequence. This imply that now we have a word of length $\ell$ in which $t$
positions have unknown values, i.e. $t$ erasures.
Moreover, $\C_{q_2}(\ell,t)$ can correct up to $t$ erasures. Hence, using the decoder
of $\C_{q_2}(\ell,t)$, the codeword $\bfc$ can be recovered from $\bfc(\Delta^-)$ and $\Delta^-$.

Thus, the codeword $\bff=(f_1,\ldots,f_{\ell})$ such that $f_i=(c_i,s_i)$ for $1\leq i \leq \ell$,
can be recovered from $\bfc$ and $\bfs$, and the theorem is proved.
\end{proof}

The proof of Theorem \ref{thm:construct_blkofdels} implies a simple decoding algorithm to recover $\bff$.
Let
\[R(\C_{q_2}(\ell,t)) = \frac{\log_{q_2} |\C_{q_2}(\ell,t)|}{\ell}\]
denote the rate of the code $\C_{q_2}(\ell,t)$.
By Theorem~\ref{thm:k=1} there exists a $(b,1)$-constrained de Bruijn sequences $\pi$ over an alphabet
of size $q_1$ if $q_1 \geq b$. By Construction~\ref{con:blkofdels}, if $\bfs$ is a $(b,1)$-constrained de Bruijn sequence
of length $\ell$ and there exists a $q_2$-ary $t$-erasure-correcting code of length $\ell$,
then $|\C_q(b,\ell,t)|=|\C_{q_2}(\ell,t)|$. If $q_1=q/q_2=b$, then the rate of the $q$-ary code $\C_q(b,\ell,t)$ is
\begin{align*}
 &R=\frac{\log_q|\C_q(b,\ell,t)|}{\ell} = \frac{\log_q q_2 \cdot \log_{q_2}|\C_{q_2}(\ell,t)| }{\ell}\\
 &=\log_q (q/q_1) \cdot R(\C_{q_2}(\ell,t))=(1-\log_q b)\cdot R(\C_{q_2}(\ell,t)).
\end{align*}
It is well-known that a code with minimum Hamming distance $t+1$ can correct $t$ erasure errors.
Moreover, for any  $0<\epsilon,\delta<1$, there exists a ``near MDS" code~\cite{AEL95} of length $\ell$ with
minimum Hamming distance $t=\delta \cdot \ell$ and rate $R(\C_{q_2}(\ell,t))\geq 1-\delta-\epsilon$ \cite{GI04}.
Hence, there exists a code of length $\ell$ correcting $t=\delta \cdot \ell$ erasures
whose rate is $R(\C_{q_2}(\ell,t))\geq 1-\delta-\epsilon$.
Therefore, we have the following theorem.

\begin{theorem}
\label{thm:rate_blksofdels}
Given $0 < \delta,\epsilon <1,$ there exists a $q$-ary $b$-limited $t_1$-burst-deletion-correcting code
of length~$\ell$ such that its rate $R$ satisfies
$$
R \geq (1-\log_q (b+1)) \cdot (1-\delta-\epsilon)~,
$$
where  $t_1 \cdot b=\delta \cdot \ell$.
\end{theorem}

The next goal is to study error-correcting codes to combat both
under shift and limited-over-shift errors. For this purpose we start by
generalizing Lemma~\ref{lem:construct_blkofdels}.

\begin{lemma}
\label{lem:combi}
Let $\bfs$ be a $(b,1)$-constrained de Bruijn sequence over an alphabet of size $q$.
Let $\bfs(\Delta^-,\Delta^+)$ be the sequence obtained from $\bfs$ after deleting
symbols in locations specified by $\Delta^-=\{i_1,\ldots,i_t\}$ and insertion
of symbols in locations specified by $\Delta^+$.
Assume further that $i_1 < \cdots < i_t$
and there are at most $b-2$ consecutive numbers in $\Delta^-$.
Then the sets $\Delta^-$ and $\Delta^+$ are uniquely determined from $\bfs$ and $\bfs(\Delta^-,\Delta^+)$.
\end{lemma}
\begin{proof}
Since $\bfs$ is a $(b,1)$-constrained de Bruijn sequence, it follows that between any two equal symbols there are
at least $b-1$ different symbols. Hence, all the sticky insertions can be located and corrected.
Now, Lemma~\ref{lem:construct_blkofdels} can be applied to find the positions of the deletions,
i.e. to determine $\Delta^-$. Since, the locations of the sticky insertions are already known,
it follows that $\Delta^-$ can be now determined.
\end{proof}

It is now straightforward to generalize Theorem~\ref{thm:construct_blkofdels}.
An $q$-ary $b$-limited $t_1$-burst-deletion sticky-insertion-correcting code is a code over $\F_q$
which corrects any number of sticky insertions and $t_1$ bursts of deletions, where each burst has length at most $b$.

\begin{theorem}
\label{thm:combi}
The code $\C_q(b,\ell,t)$ obtained in Construction~\ref{con:blkofdels} is a $q$-ary $(b-2)$-limited $t_1$-burst-deletion sticky-insertion-correcting code where $t=t_1 \cdot (b-2)$.
\end{theorem}

\begin{corollary}
\label{cor:main_blksofdels}
Given $0 < \delta,\epsilon < 1,$ there exists a $q$-ary $b$-limited
$t_1$-burst-deletion sticky-insertion-correcting code of length $\ell$, where
$t_1 \cdot b =\delta \cdot \ell$ whose rate $R$ satisfies
$$
R \geq (1-\log_q (b+2)) \cdot (1-\delta-\epsilon).
$$
\end{corollary}
It is clear that an upper bound on the maximum rate of our codes is at most $1-\delta$,
Since $\epsilon$ is arbitrarily small, when $b$ is small
and $q=2^m$ is large, it follows that the rates of our codes are close to the upper bounds and
hence they are asymptotically optimal.

\subsection{An Acyclic $\ell$-symbol Read Channel}

In this subsection, we consider a slight modification of the $\ell$-symbol read channel, where
the symbols are not read cyclically. This acyclic $\ell$-symbol read channel will be used in subsection~\ref{sec:Extraheads}
to correct synchronization errors in the racetrack memories with extra heads.

\begin{definition}
Let $\bfx=(x_1,x_2,\ldots,x_n) \in \F_q^n$ be a $q$-ary sequence of length $n$. In the
\emph{acyclic $\ell$-symbol read channel}, if $\bfx$ is the codeword then the corresponding $\ell$-symbol
read sequence of $\bfx$ is
$$
\pi_{\ell}(\bfx)=((x_1,\ldots,x_{\ell}), (x_2,\ldots,x_{\ell+1}),\ldots,(x_{n},\ldots,x_{n+\ell-1})),
$$
where $x_i$ is chosen arbitrarily for $i > n$.
\end{definition}

In this channel, as in the cyclic $\ell$-symbol read channel two type of
synchronization errors can occur, bursts of deletions where the length of each one
is restricted to at most $b-2$ deletions and sticky insertions.

\begin{example}
Let $\bfx=(0,1,0,0,1,0,0,0)$ be a stored information word. When $\ell=2$, the 2-symbol read sequence of $\bfx$ is
$$
\pi_2(\bfx)=((0,1),{\color{red}(1,0)},(0,0),(0,1),(1,0),{\color{blue}(0,0),(0,0)},(0,*)),
$$
where $*$ can be any value.
If there is a sticky insertion at the second location and a burst of deletions of length two at the 6-th
and 7th locations, we obtain the 2-symbol read sequence
$$
\pi_2(\bfy)=((0,1),{\color{red}(1,0),(1,0)},(0,0),(0,1),(1,0),(0,*))~.
$$
\end{example}

\begin{theorem}
\label{thm:acyclic_l_symbol}
Let $b$ and $k$ be two positive integers such that ${q^k \geq b\geq 2}$, where $\ell=k+b-2$. If
the first $\ell$-symbol read has no errors, then the code $\C_{DB}(n,b,k)$
can correct any number of sticky-insertions and any number of bursts of deletions of length at most $b-2$
in the acyclic $\ell$-symbol read channel.
\end{theorem}

The proof of Theorem~\ref{thm:acyclic_l_symbol} is identical to the proof of Theorem~\ref{thm:cyclic_l_symbol}
with only one distinction. In the cyclic channel we used the relation between the last $\ell$-read and the first
$\ell$-read to determine where they overlap. This cannot be done in the acyclic case. Instead, we assummed that the
first $\ell$-read has no error. Hence, we can start generating the codeword from the first symbol and thus we
can complete its decoding.

\subsection{Correcting errors in Racetrack Memories with Extra Heads}
\label{sec:Extraheads}

In this subsection, we present our last application for $(b,k)$-constrained de Bruijn codes
in the construction of codes correcting shift-errors in racetrack memories.

Let $N,n,m$ be three positive integers such that $N=n\cdot m.$
The racetrack memory comprises of $N$ cells  and $m$ heads which are uniformly distributed.
Each head will read a segment of $n$ cells. For example, in Fig.~\ref{figure:1}, the racetrack memory contains 15 data cells
and three heads are placed initially at the positions of cells $c_{1,1}$, $c_{2,1}$, and $c_{3,1}$,
respectively. Each head reads a data segment of length~5.

\begin{figure*}[t!]
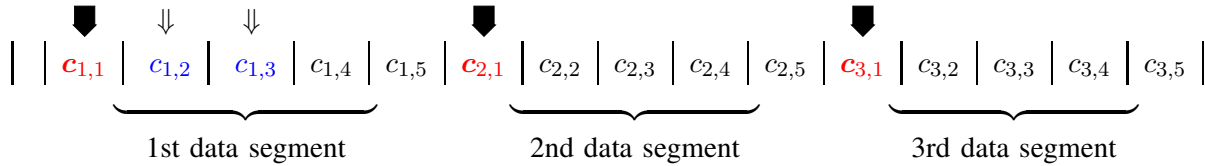

\centering
\renewcommand{\arraystretch}{1.2}
\begin{tabular}{|c|c|c|c|c|c|c|c|c|c|c|c|c|c|c|c|}
\multicolumn{1}{c}{}&\multicolumn{1}{c}{$\arrowdown$} &\multicolumn{1}{c}{$\Downarrow$} &\multicolumn{1}{c}{$\Downarrow$} &\multicolumn{2}{c}{}&  \multicolumn{1}{c}{$\arrowdown$} &\multicolumn{4}{c}{}& \multicolumn{1}{c}{$\arrowdown$}\\
  &\color{red}{$\bfc_{1,1}$}&\color{blue}{ $c_{1,2}$}&\color{blue}{ $c_{1,3}$}&$c_{1,4}$& $c_{1,5}$&\color{red}{$\bfc_{2,1}$}&$c_{2,2}$&$c_{2,3}$&$c_{2,4}$&$c_{2,5}$&\color{red}{$\bfc_{3,1}$}&$c_{3,2}$&$c_{3,3}$&$c_{3,4}$&$c_{3,5}$\\
 \multicolumn{1}{c}{} &  \multicolumn{5}{c}{\upbracefill}&\multicolumn{5}{c}{\upbracefill}&\multicolumn{5}{c}{\upbracefill}\\
  \multicolumn{1}{c}{} &  \multicolumn{5}{c}{$\text{1st data segment}$}&\multicolumn{5}{c}{$\text{2nd data segment}$}&\multicolumn{5}{c}{$\text{3rd data segment}$}\\

\end{tabular}

\caption{Racetrack memory with two extra heads}
\label{figure:1}
\end{figure*}

In general, if $\bfc=(c_1,c_2,\ldots,c_N)$ is the stored data
then the output of the $i$-th head is
$\bfc^i=(c_{i,1},\ldots,c_{i,\ell})$ where $c_{i,j}=c_{(i-1)\cdot n+j}$ for $1 \leq i \leq m$ and $1 \leq j \leq n$. Hence, an output matrix describing the output from all $m$ heads (without error) is:
\[
\begin{pmatrix}
\bfc^1\\
\bfc^2\\
\vdots \\
\bfc^m
\end{pmatrix}
=
\begin{pmatrix}
c_{1,1} & c_{1,2} & \ldots & c_{1,n}\\
c_{2,1} & c_{2,2} & \ldots & c_{2,n}\\
\vdots & \vdots & \ddots & \vdots \\
c_{m,1} & c_{m,2} &\ldots & c_{m,n}
\end{pmatrix}.
\]

When an under-shift error (a sticky-insertion error) occurs, one column is added to the above matrix by repeating a related
column of the matrix. When over-shift errors (a deletion error or a burst of $b$ consecutive deletions) occur, one or a few consecutive columns in
the matrix are deleted. Our goal is to combat these shift errors in racetrack memories.
We note that each column in the above matrix can be viewed as a symbol in an alphabet of size $q=2^m$.
In particular, let $\Phi_{m}: \F_2^{m} \mapsto \F_{q}$ be any bijection.
For each column $\whbfc_j=(c_{2,j},\ldots,c_{m,j})^T$, $\Phi_{m}(\whbfc_j) = v_j \in \F_{q}$.
Hence, by Corollary \ref{cor:main_blksofdels} , it is straightforward to obtain the following result.
\begin{proposition}\label{prop.no.extra.head}
Given $0 < \delta, \epsilon <1$, there exists a code correcting any number of under-shift errors
and $t_1$ over-shift errors, each of length at most $b$, with rate
$R \geq \frac{m-\log_2 (b+2)}{m} (1-\delta-\epsilon).$
\end{proposition}

Another way to combat these type of errors is to add
some consecutive extra heads next to the first head. For example, in Fig. \ref{figure:1},
there are two extra heads next to the first head. We assume in this section that
there are $\ell-1$ extra heads. Since there are two types of heads, we call the $\ell-1$ extra
heads \emph{secondary heads}, while the first $m$ uniformly distributed heads are the \emph{primary heads}.
Hence, there are $\ell$ heads which read the first data segment together, the first
primary head and all the $\ell-1$ secondary heads. For $2 \leq i \leq m$, each other primary head
will read one data segment individually.  The output from the last $(m-1)$ primary heads
is $ \bfc[n+1,N]=(\bfc^2,\ldots,\bfc^m)$, where
\[
\begin{pmatrix}
\bfc^2\\
\vdots \\
\bfc^m
\end{pmatrix}
=
\begin{pmatrix}
c_{2,1} & c_{2,2} & \ldots & c_{2,n-1} &c_{2,n}\\
\vdots & \vdots & \ddots & \vdots & \vdots \\
c_{m,1} & c_{m,2} &\ldots & c_{m,n-1}& c_{m,n}
\end{pmatrix}.
\]
This matrix can be viewed as a $q_2$-ary word of length $n$ where each column is
a symbol in the alphabet of size $q_2=2^{m-1}$. In particular, let $\Phi_{m-1}: \F_2^{m-1} \mapsto \F_{q_2}$ be any bijection.
For each column $\whbfc_j=(c_{2,j},\ldots,c_{m,j})^T$, $\Phi_{m-1}(\whbfc_j) = v_j \in \F_{q_2}$,
and define
$$
\Phi (\bfc[n+1,N]) \triangleq (\Phi_{m-1}(\whbfc_1),\Phi_{m-1}(\whbfc_2),\ldots,\Phi_{m-1}(\whbfc_n)) =(v_1,\ldots,v_{n})=\bfv \in \F^n_{q_2}~.
$$

The output from all the $\ell$ heads in the first segment is

\[
\begin{pmatrix}

\bfc^{1,1}\\
\bfc^{1,2}\\
\vdots \\
\bfc^{1,\ell}
\end{pmatrix}
=
\begin{pmatrix}
c_{1,1} & c_{1,2} & \ldots & c_{1,n-1} &c_{1,n}\\
c_{1,2} & c_{1,3} & \ldots & c_{1,n} & *\\
\vdots & \vdots & \ddots & \vdots & \vdots \\
c_{1,\ell} & c_{1,\ell+1} &\ldots & *& *
\end{pmatrix}.
\]
It is readily verified that this is the acyclic $\ell$-symbol read sequence $\pi_{\ell}(\bfc[1,n])$
for the first data segment $\bfc[1,n]=(c_{1,1},\ldots,c_{1,n})$. This motivates the following construction.

\begin{construction}
\label{constr.extra1}
Let $m,n,b,k,\ell$ be positive integers such that $\ell=k+b-2$,
let $\bfc^1=\bfc[1,n]= (c_{1},\ldots,c_{n})\in \C_{DB}(n,b,k)$, let
$\C_{q_2}(n,t)$ be a $q_2$-ary $t$-erasure-correcting code of length $n$,
where ${t=(b-2)\cdot t_1}$, for some integer $t_1$ and $q_2=q^{m-1}$, and let
$\bfc[n+1,N] =(\bfc^2,\ldots,\bfc^m)$, where
$$
\Phi (\bfc[n+1,N]) \triangleq (\Phi_{m-1}(\whbfc_1),\Phi_{m-1}(\whbfc_2),\ldots,\Phi_{m-1}(\whbfc_n)) =(v_1,\ldots,v_{n})=\bfv \in \C_{q_2}(n,t)~.
$$
Define,
$$
\C_3 (N,t_1,b-2) \triangleq \{ ( \alpha_1 ,\alpha_2) ~:~ \alpha_1 \in \bfc[1,n],~\alpha_2 \in \bfc[n+1,N]\}.
$$
\end{construction}

\begin{theorem}\label{thm:size:extra}
The code $\C_3 (N,t_1,b-2)$ has size $|\C_{DB}(n,b,k)| \cdot |\C_{q_2}(n,t)|$
and using $\ell -1$ extra heads it can correct any number of sticky insertions and $t_1$ bursts
of deletions whose length is at  most $b-2$.
\end{theorem}
\begin{proof}
The size of the code is an immediate observation from the definition of the code $\C_3 (N,t_1,b-2)$.
The first data segment of length $n$ consists of any sequence $\bfs$ from $\C_{DB}(n,b,k)$.
By Theorem \ref{thm:acyclic_l_symbol},
we can recover the sequence in the first data segment when there are any number sticky-insertions
and bursts of deletions of length at most $b-2$. Moreover, we can also determine the locations of these errors.
In the output from the last $m-1$ heads, all sticky insertions can be corrected easily
and all deletions become erasures since we know the locations of these errors.
There are at most $t = t_1 \cdot (b-2)$ erasures and the decoding procedure of the code $\C_{q_2}(n,t)$ can
be applied to correct these erasures. Thus, the sequence $\bfs$ can be recovered.
Thus, the code ${\C_3(N,t_1,b-2)}$ can correct
all sticky insertions and at most $t_1$ bursts of deletions whose length at most $b-2$.
\end{proof}

\begin{corollary}
\label{cor:main.extra1}
Consider a racetrack memory comprising of ${N=m\cdot n}$ cells and $m$ primary heads which are uniformly distributed.
Using $\ell-1$ extra secondary heads, it is possible to construct a code correcting a combination
of any number of sticky insertions and $t_1$ bursts of deletions whose length is at most $b-2$ such that its asymptotic rate satisfies
$$
\lim_{N \to \infty } \frac{\log_2 |\C_3 (N,t_1,b-2)|}{N} \geq \frac{m-1}{m} \cdot (1-\delta-\epsilon) + \frac{R_{DB}(b,k)}{m}
$$
where $\ell-b+2=h$,
$t_1 \cdot (b-2) =\delta \cdot n$ and $0< \delta, \epsilon <1$.
\end{corollary}

\begin{proof}
We note that there exists a code $\C_{q_2}(n,t)$ of length $n$ correcting $t= t_1 \cdot (b-2)$ erasures
with the asymptotic rate at least
$$
\lim_{n \to \infty} \frac{\log_{q_2} |\C_{q_2}(n,t)|}{n} \geq 1-\delta-\epsilon.
$$
Since $q_2=2^{m-1}$ and $N=mn$, we have
\begin{equation}
\label{eq11}
\lim_{n \to \infty} \frac{\log_{2} |\C_{q_2}(n,t)|}{N} \geq \frac{m-1}{m}(1-\delta-\epsilon).
\end{equation}
Moreover, the asymptotic rate of the constrained de Bruijn code is
\begin{equation}\label{eq22}
R_{DB}(b,k)=\lim_{n \to \infty} \frac{\log_2 |\C_{DB}(n,b,k)|}{n}.
\end{equation}
Therefore, by \eqref{eq11},  \eqref{eq22}, and Theorem \ref{thm:size:extra},
the rate of the code $\C_3(N,t_1,b-2)$ in construction \ref{constr.extra1} can
be computed as follows:
$$
\lim_{N \to \infty} \frac{\log_2 |\C_3(N,t_1,b-2)| }{N}= \lim_{N \to \infty} \frac{\log_2 (|\C_{q_2}(n,t)| |\C_{DB}(n,b,k)|) }{N} \geq \frac{m-1}{m}(1-\delta-\epsilon)+\frac{R_{DB}(b,k)}{m}.
$$
\end{proof}

Corollary~\ref{cor:main.extra1} can be compared with the result in Proposition \ref{prop.no.extra.head}.
When $b=h=3$, by Table \ref{table1},
we have that $R_{DB}(3,3) \approx 0.7946$. Hence, using two more extra heads, the asymptotic rate of the relate code is $\frac{m-1}{m}(1-\delta-\epsilon)+\frac{0.7946}{m}= 1-\delta-\epsilon + \frac{0.7946-1+\delta+\epsilon}{m}$. We note that, without using extra head, the maximal asymptotic rate is $1-\delta$. Hence, when $1-\delta < 0.7946$, using two extra heads,
the asymptotic rate of our constructed code is higher than the maximal asymptotic rate of codes without extra heads.




\section{Conclusions and Open Problems}
\label{sec:conclude}

We have defined a new family of sequences and codes named constrained de Bruijn sequences,
This family of codes generalizes the family of de Bruijn sequences.
These new defined sequences have
some constraints on the possible appearances of the same $k$-tuples in
substrings of a bounded length.
As such these sequences can be viewed also as constrained sequences and
the related codes as constrained codes.
Properties and constructions of such sequences, their enumeration, encoding and decoding for the
related codes, were discussed.
We have demonstrated applications of these sequences for combating synchronization errors
in new storage media such as the $\ell$-symbol read channel and the racetrack memories.
The new defined sequences raise lot of questions for future research from which we outline a few.

\begin{enumerate}
\item Find more constructions for constrained de Bruijn codes with new parameters
and with larger rates. Especially, we are interested in $(b,k)$-constrained de Bruijn
codes for which $b$ is about $q^t$ and $k= c \cdot t$, where $c$ is a small constant,
and the rate of the code is greater than zero also when $k$ go to infinity.

\item Find better bounds (lower and upper) on the rates of constrained de Bruijn codes with various parameters.
Especially we want to find the exact rates for infinite families of parameters, where each family
itself has an infinite set of parameters.

\item What is the largest number of de Bruijn sequences of order $k$ over $\F_q$ such that
the largest substring that any two sequences share is of length at most $k+\delta$, where $2 \leq \delta \leq k-1$.

\item Find more applications for constrained de Bruijn sequences and constrained de Bruijn codes.
\end{enumerate}


%


\end{document}